\documentclass[10pt,a4paper]{article}

\usepackage[utf8]{inputenc}
\usepackage[T1]{fontenc}
\usepackage{microtype}
\usepackage{braket}
\usepackage{fullpage}

\usepackage{tikz}
\usepackage{amsmath}
\usepackage{amsthm}
\usepackage{amssymb}
\usepackage{mathscinet}
\numberwithin{equation}{section}

\usepackage[nottoc,notlot,notlof]{tocbibind}
\usepackage[title]{appendix}

\usepackage{aliascnt}
\usepackage[colorlinks,linkcolor=blue,citecolor=blue]{hyperref}

\theoremstyle{plain}

\newtheorem{theorem}{Theorem}[section]

\newaliascnt{lemma}{theorem}
\newtheorem{lemma}[lemma]{Lemma}
\aliascntresetthe{lemma}

\newaliascnt{proposition}{theorem}
\newtheorem{proposition}[proposition]{Proposition}
\aliascntresetthe{proposition}

\newaliascnt{corollary}{theorem}
\newtheorem{corollary}[corollary]{Corollary}
\aliascntresetthe{corollary}

\theoremstyle{definition}

\newaliascnt{definition}{theorem}
\newtheorem{definition}[definition]{Definition}
\aliascntresetthe{definition}

\newaliascnt{example}{theorem}
\newtheorem{example}[example]{Example}
\aliascntresetthe{example}

\newaliascnt{remark}{theorem}
\newtheorem{remark}[remark]{Remark}
\aliascntresetthe{remark}

\newcommand{\RR}{\mathbf{R}}
\newcommand{\C}{\mathbf{C}}
\newcommand{\D}{\mathcal{D}}

\renewcommand{\epsilon}{\varepsilon}
\def\D{{\mathcal D}}
\def\H{{\mathcal H}}
\def\E{{\mathcal E}}

\def\inv{{}^{-1}}
\def\X{{\rm X}}
\def\inum{{\rm i}}

\DeclareMathOperator{\re}{Re}
\DeclareMathOperator{\im}{Im}

\newcommand{\abs}[1]{\left\lvert #1 \right\rvert}
\newcommand{\avg}[1]{\bigl\langle #1 \bigr\rangle}

\DeclareMathOperator{\sgn}{sgn}

\title{The dynamics of conservative peakons in the NLS hierarchy }
\author{Stephen C. Anco
\thanks{Department of Mathematics and Statistics, Brock University, 
St. Catharines, ON L2S3A1, Canada; sanco@brocku.ca}
\and Xiangke Chang
\thanks{LSEC, ICMSEC, Academy of Mathematics and Systems Science, Chinese Academy of Sciences, P.O.Box 2719, Beijing 100190, PR China and School of Mathematical Sciences, University of Chinese Academy of Sciences, Beijing 100049, PR China, and Department of Mathematics and Statistics, University of Saskatchewan, 106 Wiggins Road, Saskatoon, Saskatchewan, S7N 5E6, Canada; changxk@lsec.cc.ac.cn}
\and Jacek Szmigielski
\thanks{Department of Mathematics and Statistics, University of Saskatchewan, 106 Wiggins Road, Saskatoon, Saskatchewan, S7N 5E6, Canada; szmigiel@math.usask.ca}}
\date{\today}

\makeatletter
\hypersetup{%
  pdfauthor={Jacek Szmigielski},
  pdftitle={\@title},
}
\makeatother

\begin{document}
\maketitle
\begin{abstract} 
Using the tri-hamiltonian splitting method, 
the authors of \cite{anco} derived two $U(1)$-invariant nonlinear PDEs that arise from the hierarchy of the nonlinear Schr\"{o}dinger equation 
and admit peakons (\textit{non-smooth solitons}).
In the present paper, 
these two peakon PDEs are
generalized to a family of $U(1)$-invariant peakon PDEs parametrized by 
the real projective line $\RR \mathbf{P} ^1$.  
All equations in this family are shown to posses  
\textit{conservative peakon solutions} 
(whose Sobolev $H^1(\RR)$ norm is time invariant).   
The Hamiltonian structure for the sector of conservative peakons is identified 
and the peakon ODEs are shown to be Hamiltonian with respect to several Poisson structures.  
It is shown that the resulting Hamilonian peakon flows 
in the case of the two peakon equations derived in \cite{anco} 
form orthogonal families, while in general the Hamiltonian peakon flows
for two different equations in the general family
intersect at a fixed angle equal to the angle between two lines  in $\RR \mathbf{P} ^1$ parametrizing those two equations.  
Moreover, it is shown that inverse spectral methods 
allow one to solve explicitly the dynamics of conservative peakons 
using explicit solutions to a certain interpolation problem.  
The graphs of multipeakon solutions confirm the existence of 
multipeakon breathers as well as asymptotic formation of pairs of 
two peakon bound states in the non-periodic time domain.  
\end{abstract} 


\tableofcontents{}
\section{Introduction} 
The origin of the present paper goes back at least to the 
fundamental paper by R.Camassa and D.Holm \cite{ch} 
in which they proposed the nonlinear partial differential equation 
\begin{equation} \label{eq:CH}
m_t+u m_x +2u_x m+2\kappa u_x=0, \qquad  m=u-u_{xx}, 
\end{equation} 
as a model equation for the shallow water waves.  
The coefficient $\kappa$ appearing in \eqref{eq:CH} is proportional to a critical shallow water wave speed.  
One of the results of \cite{ch} was that 
\eqref{eq:CH} has soliton solutions which 
are no longer smooth in the limit of $\kappa \rightarrow 0^+$.
These non-smooth solitons have the form of \textit{peakons}
\begin{equation} \label{eq:peakonansatz}
u=\sum_{j=1}^n m_j (t)e^{-\abs{x-x_j(t)}}, 
\end{equation} 
where all coefficients $m_j(t)$ and the positions $x_j(t)$ are assumed to 
be smooth as functions of $t$.  
It is elementary to 
see that $m=u-u_{xx}$ becomes
a discrete measure $m=2\sum_j m_j \delta_{x_j}$ 
and that when the CH equation \eqref{eq:CH} for $\kappa=0$ is interpreted in an appropriate weak sense, it
turns into a system of Hamilton's equations of motion
\begin{equation} \label{eq:HCHpeakons}
\dot x_j=\{x_j, H\}, \qquad \dot m_j=\{m_j, H\}, 
\end{equation} 
with the Hamiltonian $H=\tfrac12 \sum_{i,j} m_i m_j  e^{-\abs{x_i-x_j}}$
and with $x_j, m_j$ being canonical conjugate variables.  
This system is Lax integrable in the sense 
that it can be written as a matrix Lax equation \cite{ch}.  
Moreover, 
there exists an explicit solution in terms of Stieltjes continued fractions 
\cite{bss-Stieltjes}.  The CH equation has generated over the years 
a remarkably strong response from the scientific community, attracted to its 
unique features such as the breakdown of regularity of 
its solutions \cite{constantin-escher, McKean-Asianbreakdown, McKean-breakdown, bss-moment} and the stability of its solutions, including the 
peakon solutions 
\cite{constantin-strauss, molinet-n-peakons, molinet-multipeakons}.  
One of the outstanding issues that has emerged over the last two decades 
has been the question of understanding, or perhaps even classifying, 
equations sharing these distinct features of the CH equation.   One of the earliest results in this direction 
was obtained by P. Olver and P. Rosenau in 
\cite{olver-rosenau-triH}  who put forward 
an elegant method, called by them a tri-Hamiltionian duality, which 
provided an intriguing way of deriving not only the CH equation but also  
other equations possessing peakon solutions.  The main idea of that paper 
can be illustrated, 
as  in their paper,
on the example of 
the Korteweg-de Vries (KdV) equation
\begin{equation} \label{eq:KdV}
u_t=u_{xxx}+3uu_x.  
\end{equation}
This equation is known to have a bi-Hamiltonian structure.  
Indeed, using the two compatible Hamiltonian operators 
\begin{equation} \label{eq:KdV J}
J_1=D_x, \qquad J_2 =D_x^3+uD_x+D_xu, 
\end{equation} 
and the Hamiltonians 
\begin{equation} \label{eq:KdV H}
H_1=\tfrac 12 \int u^2 dx, \qquad H_2=\tfrac 12 \int (-u_x^2+u^3) dx, 
\end{equation} 
the KdV equation \eqref{eq:KdV} can be written as
\begin{equation} \label{eq:KdV-biH}
u_t=\{u,H_1\}_{J_2}=\{u,H_2\}_{J_1}, 
\end{equation} 
where the Poisson bracket associated with the Hamiltonian operator 
$J$ is given by $\{f,g\}_J=\int \frac{\delta f}{\delta u} J \frac{\delta g}{\delta u}dx$ for smooth in $u$ functionals $f,g$.  The decisive step is now 
to redistribute the Hamiltonian operators by 
splitting and reparametrizing differently the compatible triple 
$D_x, D_x^3, uD_x+D_xu$ of Hamiltonian operators 
to create a new compatible pair
\begin{equation}\label{eq:CH J}
\hat J_1 =D_x-D^3_x=D_x\Delta , \qquad \hat J_2=mD+Dm, 
\end{equation} 
with $m=(1-D_x^2)u=\Delta u$.
The appearance of the factorization of one of the Hamiltonian operators in 
terms of the Helmholtz operator $\Delta$ is of paramount importance 
for creating ``peakon'' equations.  In particular, 
the CH equation \eqref{eq:CH} can be written now in bi-Hamiltonian form 
\begin{equation} \label{eq:CH-biH}
m_t=\{m,\hat H_1\}_{\hat J_2}=\{m,\hat H_2\}_{\hat J_1}, 
\end{equation} 
with the Hamiltonians 
\begin{equation}\label{eq:CH H}
\hat H_1=\tfrac 12 \int u m dx, \qquad \hat H_2=\tfrac 12 \int (-uu_x^2+u^3) dx.  
\end{equation} 
It is shown in \cite{olver-rosenau-triH} that applying this methodology to 
the modified KdV equation one obtains 
the nonlinear partial differential equation
\begin{equation}\label{eq:m1CH}
m_t+\left((u^2-u_x^2) m\right)_x=0,  \qquad 
m=u-u_{xx}. 
\end{equation}
\autoref{eq:m1CH} appeared also in the papers of T. Fokas \cite{fokas1995korteweg} and B. Fuchssteiner \cite{fuchssteiner1996some},   and was, later, rediscovered by Z. Qiao  
 \cite{qiao2006new,qiao2007new}. Some early work 
 on the Lax formulation of this equation was done by J. Schiff \cite{schiffdual} . Recently this equation has attracted a considerable attention from many authors
 \cite{kang-liu-olver-qu, gui2013wave,liu2014orbital,Himonas4peak,himonas2014cauchy, chang-szmigielski-m1CHlong}. 
  
Interestingly enough, the same philosophy applied to the 
non-linear Schr\"{o}dinger equation 
\begin{equation}\label{eq:NLS}
u_t=i(u_{xx}+\abs{u}^2 u)  
\end{equation}
appeared to produce no peakon equations.  
We recall that the bi-Hamiltonian formulation of \eqref{eq:NLS} 
which was 
used in  \cite{olver-rosenau-triH} 
is 
based on
the standard NLS Hamiltonian operators
\begin{equation} \label{eq:NLS J}
J_1(F)=\inum F,   \qquad J_2(F)=D_xF+uD_x^{-1}(\bar u F-u\bar F), 
\end{equation} 
written in action on densities $F$.  
The two redistributed Hamiltonian operators 
\begin{equation*} 
\hat J_1(F)=(D_x+i)F, \qquad \hat J_2(F)=mD_x^{-1}(\bar m F-m\bar F)
\end{equation*} 
do not contain the Helmholtz operator $\Delta$ in their factorizations.  
This puzzling situation was resolved in the paper \cite{anco} by 
S. Anco and F. Mobasheramini in which the authors proposed 
a different choice of Hamiltonian operators resulting 
in the bi-Hamiltonian formulation of {\bf two }new peakon equations 
\begin{equation}\label{eq:HPorig} 
m_t +((|u|^2-|u_x|^2)m)_x +2\inum\im(\bar u u_x)m =0, 
\end{equation}
called the Hirota-type (HP) peakon equation, and 
\begin{equation}\label{eq:NLSPorig} 
\inum m_t +2\inum(\im(\bar u u_x)m)_x + (|u|^2-|u_x|^2)m =0, 
\end{equation}
called the NLS-type (NLSP) peakon equation. 
Unlike other peakon equations, 
both the HP and NLSP equations display the same $U(1)$-invariance as the NLS equation.

In the remainder of this introduction, 
we outline the content of individual sections, highlighting the main results. 
 
\noindent To begin,
in \autoref{sec:HamNLS} we review the Hamiltonian 
setup for both HP and NLSP following \cite{anco}. 
 
\noindent In \autoref{sec:Laxpair} we give a unifying perspective 
on the Lax pairs for \eqref{eq:HPorig} and \eqref{eq:NLSPorig}, 
showing that these equations are just two members of a family of 
peakon equations parametrized by the real projective line $\RR \mathbf{P}^1$.  
  
\noindent In \autoref{sec:peakons} we introduce and study \textit{conservative} 
peakons obtained from the distributional formulation of 
the Lax pairs discussed in \ref{sec:Laxpair}.  
This type of peakon solutions not only preserves 
the Sobolev $H^1$ norm but also admits multiple Hamiltonian formulations 
which we study in detail, giving in particular a unifying 
Hamiltonian formulation for the whole family of conservative peakon equations.  

\noindent In \autoref{sec:HPspectral} we concentrate on the isospectral boundary 
value problem relevant for the peakon equation \eqref{eq:pHP}. 

\noindent In \autoref{sec:IP} we develop a two-step procedure for solving 
the inverse problem for the HP equation which, effectively, 
recovers the peakon measure \eqref{eq:peakon measure} in terms of solutions to an interpolation problem stated in \autoref{thm:CJ interp}.  
We provide graphs of solutions obtained from explicit formulas 
and give preliminary comments about the dynamics of peakon solutions.

\section{Hamiltonian structure of NLS and Hirota peakon equations}
\label{sec:HamNLS}

This section is a condensed summary of 
the part of \cite{anco} relevant to the present paper.  

Hamiltonian structures go hand in hand with Poisson brackets. 
In particular, 
a linear operator $\E$ is a Hamiltonian operator iff 
its associated bracket 
\begin{equation} \label{eq:PB}
\{F,G\}_{\E} =\braket{ \delta F/\delta\bar{m}, \E(\delta G/\delta\bar{m})}
\end{equation} 
on the space of functionals of $(m,\bar m)$ is a Poisson bracket,
namely, this bracket is skew-symmetric and obeys the Jacobi identity, 
turning the space of functionals into a Lie algebra.  
Here $\braket{f,g}$ is a real, symmetric, bilinear form defined as 
\begin{equation}
\braket{f,g}=\int_{\RR} (\bar f(x) g(x)+f(x)\bar g(x)) \, dx, 
\end{equation}
thus equipping the space of $L^2$ complex functions with a real, positive-definite, inner product. 
Note that if $F,G$ are real functionals, then the bracket \eqref{eq:PB} takes the form 
\begin{equation} 
\{F, G\}=\int_{\RR}\big((\delta F/\delta m) \, \E(\delta G/\delta\bar{m})+
(\delta F/\delta\bar{m})\,  \bar{\E} (\delta G/\delta m)\big)\, dx, 
\end{equation} 
where the variational derivative of each functional with respect to $(m,\bar{m})$ 
is defined relative to the inner product by
\begin{equation}
\delta H = \braket{\delta H/\delta m,\delta \bar{m}} = \braket{\delta H/\delta \bar{m},\delta m}
\end{equation} 
for each real functional $H$.  
In explicit form, 
the variational derivative of real functional $H=\int_{\RR} h\,dx$ is given in terms of 
the density $h$ by the relation 
\begin{equation}
\delta H/\delta m = E_m(h),
\quad
\delta H/\delta \bar{m} = E_{\bar{m}}(h) ,
\end{equation} 
where $E_v$ denotes the Euler operator with respect to a variable $v$.

Another key relationship is that a Poisson bracket and a Hamiltonian operator 
satisfy $\{m,H\}_{\E}=\E(\delta H/\delta \bar{m})$ and, symmetrically, $\{\bar{m},H\}_{\E}=\bar\E(\delta H/\delta {m})$.  

To proceed, we first introduce the 1-D Helmholtz operator 
\begin{equation}
\Delta=1-D_x^2 
\end{equation}
which connects $m$ and $u$ through 
\begin{equation} \label{eq:um} 
m=u-u_{xx}=\Delta u.  
\end{equation}
Both the HP and NLSP equations share two compatible Hamiltonian operators, 
stated in \cite{anco}, namely
\begin{equation} \label{eq:HD}
\H  = 2\inum \Delta, 
\quad
\D = 2\inum(mD\inv_x\re\bar m D_x +\inum D_xm D\inv_x\im\bar m).  
\end{equation}
Compatibility means that every linear combination $c_1\H +c_2 \D$ of these two Hamiltonian operators is a Hamiltonian operator. 
Note, compared to the operators presented in \cite{anco},
here $\H$ and $\D$ have been normalized by a factor of $2$ 
that corresponds to our choice of normalization 
for the nonlinear terms in the equations \eqref{eq:HPorig}--\eqref{eq:NLSPorig}. 

Each Hamiltonian operator \eqref{eq:HD} defines a respective Poisson bracket 
$\{F,G\}_{\H}$ and $\{F,G\}_{\D}$. 
The bi-Hamiltonian structure of the NLSP equation is given by 
\begin{equation} \label{eq:NLSPbiHamil}
m_t = \inum(|u|^2-|u_x|^2)m -2(\im(\bar u u_x)m)_x 
= \D(\delta H^{(0)}/\delta\bar m) 
= \H(\delta H^{(1)}/\delta\bar m) , 
\end{equation}
or equivalently,
\begin{equation}
m_t = \{m, H^{(0)}\}_{\D} = \{m, H^{(1)}\}_{\H} , 
\end{equation}
where 
\begin{equation}\label{eq:H0}
H^{(0)} = \int_{\RR} \re(\bar u m)\, dx \
= \int_{\RR} (|u|^2 +|u_x|^2)\,dx 
\end{equation}
and
\begin{equation}\label{eq:H1}
H^{(1)} =\int_{\RR}\big(  \tfrac{1}{4} (|u|^2-|u_x|^2)\re(\bar{u}m) + \tfrac{1}{2}\im(\bar{u}u_x)\im(\bar{u}_xm) \big) \,dx
\end{equation}
are the Hamiltonian functionals.
Both functionals $H^{(0)}$ and $H^{(1)}$ are conserved for smooth solutions $u(t,x)$ 
with appropriate decay conditions at $\abs{x}\rightarrow \infty$. 

Likewise, 
the bi-Hamiltonian structure of the HP equation is given by 
\begin{equation}\label{eq:HPbiHamil}
m_t= -((|u|^2-|u_x|^2)m)_x -2\inum\im(\bar u u_x)m
= -\D(\delta E^{(0)}/\delta\bar m) 
= \H(\delta E^{(1)}/\delta\bar m) 
\end{equation}
where 
\begin{equation}\label{eq:E0}
E^{(0)}= \int_{\RR} \im(\bar u m_x)\, dx 
= \int_{\RR} \im(u_x \bar m)\,dx 
\end{equation}
and
\begin{equation}\label{eq:E1}
E^{(1)} = \int_{\RR}\big( \tfrac{1}{4}(|u|^2 -|u_x|^2)\im(\bar{u}_xm) -\tfrac{1}{2}\im(\bar{u}u_x)\re(\bar{u}m)  \big)\,dx
\end{equation}
are the Hamiltonian functionals. 
These two functionals $E^{(0)}$ and $E^{(1)}$ are conserved for smooth solutions $u(t,x)$ 
with appropriate decay conditions. 
In terms of Poisson brackets,
the corresponding structure is 
\begin{equation} 
m_t = \{m, -E^{(0)}\}_{\D} = \{m, E^{(1)}\}_{\H} . 
\end{equation}

We now make some brief comments about the conserved Hamiltonians $H^{(0)}$ and $E^{(0)}$, going beyond the presentation in \cite{anco}. 
Recall that once a Hamiltonian operator $\E$ (or the corresponding Poisson bracket) is given, any real functional $H$ gives rise to a Hamiltonian vector field 
\begin{equation*} 
\X_H = \eta\partial_{m} + \bar{\eta}\partial_{\bar m}, \qquad \text{where } 
\eta = \E(\delta H/\delta\bar{m}) =\{m,H\}_{\E}, 
\end{equation*} 
acting on the space of densities of real functionals. 
First, using the Hamiltonian operator $\H$, 
we see that 
\begin{equation*}
\H(\delta H^{(0)}/\delta\bar{m}) = 2\inum m,
\quad
\H(\delta E^{(0)}/\delta\bar{m}) = 2 m_x
\end{equation*}
respectively, produce, after a simple rescaling, the Hamiltonian vector fields 
$\X_{\text{phas.}} = \inum m\partial_{m} - \inum\bar{m}\partial_{\bar m}$
and
$\X_{\text{trans.}} = m_x\partial_{m} +\bar{m}_x\partial_{\bar m}$. 
These two vector fields are the respective generators of 
phase rotations $(m,\bar m)\to (e^{\inum\phi}m,e^{-\inum\phi}\bar m)$
and $x$-translations $x\to x+\epsilon$,
where $\phi,\epsilon$ are arbitrary (real) constants. 
Next, by direct computation, we find that
\begin{equation}
\{H^{(0)},E^{(0)}\}_{\H}=0,
\quad 
\{H^{(0)},E^{(0)}\}_{\D}=0.  
\end{equation}
These brackets show that $\X_{\text{phas.}}$ and $\X_{\text{trans.}}$
are commuting symmetry vector fields for both the NLSP equation and the HP equation.  As a consequence, we note the following useful features of these Hamiltonians. 
\begin{remark}
Each Hamiltonian $H^{(0)}$ and $E^{(0)}$ is conserved for both the NLSP equation and the HP equation,
and these Hamiltonians are invariant under the symmetries generated by 
the commuting Hamiltonian vector fields $\X_{\text{phas.}}$ and $\X_{\text{trans.}}$. 
In addition, 
$H^{(0)}=||u||^2_{H^1}$ is the square of the Sobolev norm of $u(t,x)$. 
\end{remark}
Some insight into the meaning of $E^{(0)}$ comes from expressing these Hamiltonians
in terms of a Fourier representation 
$u(t,x)=\tfrac{1}{\sqrt{2\pi}}\int_{\RR} e^{ikx} a(k,t)\, dk$. 
A simple computation gives 
$H^{(0)} = \int_{\RR} (1+k^2)|a(k,t)|^2\, dk$
and $E^{(0)} = \int_{\RR} k(1+k^2)|a(k,t)|^2\, dk$,
which shows that the conserved density arising from $E^{(0)}$ 
is $k$ times the conserved positive-definite density given by $H^{(0)}$. 
This is analogous to the relationship between the well-known conserved energy and momentum 
quantities for the NLS equation \cite{Johnson}. 
Since $H^{(0)}$ plays the role of a conserved positive-definite energy
for both the HP and NLSP equations, 
we can thereby view $E^{(0)}$ as being a conserved indefinite-sign momentum
for these equations. 

In the the next section \autoref{sec:Laxpair} 
we will revisit the derivation of the HP and NLSP equations \eqref{eq:NLSPorig} and \eqref{eq:HPorig}
starting from a unified perspective provided by their Lax pair formulation.

\section{A unified Lax pair}
\label{sec:Laxpair}

We begin by showing how the Lax pairs in \cite{anco} for the NLSP and HP equations 
can be unified. 

For $\lambda \in \C$, consider the family of $\mathfrak{sl}(2,\C)$ matrices 
\begin{equation}\label{eq:UV}
U=\tfrac{1}{2}
\begin{bmatrix} -1 &\lambda m\\ -\lambda \bar{m}& 1 \end{bmatrix},
\quad
V=\tfrac{1}{2}
\begin{bmatrix} \sigma(2\lambda^{-2} + Q) & -2\sigma\lambda^{-1} (u-u_x) -\lambda m J\\
2\sigma\lambda^{-1}(\bar u+\bar u_x) +\lambda \bar m J& -\sigma(2\lambda^{-2} +Q) \end{bmatrix}
\end{equation}
parametrized by two complex valued functions $m$ and $u$, 
and a complex constant $\sigma$, 
where 
\begin{equation} \label{eq:Q}
Q=(u-u_x)(\bar u+\bar u_x)=|u|^2-|u_x|^2-2i\im(\bar u u_x), 
\end{equation} 
and $J$ is a complex function that we will now determine. 
\begin{remark} \label{rem:gl2}
We point out that $V$ is not uniquely determined; 
in particular we can add to it a $\lambda$ dependent multiple of the 
identity.  This becomes a necessity when boundary conditions are imposed.  
We will return to this point further into the paper. 
\end{remark} 
We impose on the pair $(U,V)$ the zero-curvature equation
\begin{equation}
U_t-V_x+[U,V]=0,  \label{eq:zerocurv}
\end{equation}
which gives 
\begin{subequations} \label{eq:zerocurv1}
\begin{align}
& m_t +(mJ)_x +(J-\sigma Q)m=0,\label{eq:mt} \\
& \bar m_t +(\bar mJ)_x -(J-\sigma Q)\bar m=0, \\
&m=u-u_{xx}. 
\end{align}
\end{subequations}
For these equations to be compatible, 
$J$ must be real and $J-\sigma Q$ must be purely imaginary,
which implies 
\begin{equation}\label{eq:J}
J=\re(\sigma Q)=\re(\sigma)(|u|^2-|u_x|^2) +2\im(\sigma)\im(\bar u u_x)
\end{equation}
and
\begin{equation*}
J-\sigma Q =-\inum\im(\sigma Q)=2\inum\re(\sigma)\im(\bar u u_x)-\inum\im(\sigma)(|u|^2-|u_x|^2)  . 
\end{equation*}
Consequently, the zero-curvature equation \eqref{eq:zerocurv} becomes
\begin{equation}\label{eq:unifiedHPNLSP}
m_t +\re(\sigma)\big( ((|u|^2-|u_x|^2)m)_x +2i\im(\bar u u_x)m \big)
+\im(\sigma)\big( 2(\im(\bar u u_x)m)_x  -i(|u|^2-|u_x|^2)m \big)
=0
\end{equation}
which is, loosely speaking, a linear combination of the HP and NLSP equations. 
In particular, the HP equation is obtained for $\sigma=1$, 
and the NLSP equation is obtained for $\sigma =i$. 
We note, however, that 
by rescaling the $t$ variable we can put
$\abs{\sigma}=1, \im(\sigma)\geq 0$. 
In this sense \autoref{eq:mt} simplifies to
\begin{equation} \label{eq:mfamily}
m_t+(\re(e^{\inum \theta} Q)m)_x-\inum \im(e^{\inum \theta}Q)m=0, 
\end{equation}
where the angle $\theta$ can be restricted to $[0,\pi)$.  
This angle has a simple geometric interpretation as being a local 
parameter in the real projective line $\RR\mathbf{P}^1$.  
Thus different equations in this family of PDEs correspond to different points 
in $\RR \mathbf{P}^1$; in practice though 
different equations are obtained by rigid rotations
of $Q$ by the angle $\theta$ (see more on this item below). 

We conclude that the unified equation \eqref{eq:mfamily} is, at least formally, a Lax integrable system, as it arises from a Lax pair \eqref{eq:UV}, \eqref{eq:Q}, with $J$ given by \eqref{eq:J}
and $\sigma=e^{\inum \theta}$.  
The unified equation also possesses a bi-Hamiltonian structure given by 
the Hamiltonian operators \eqref{eq:HD} shared by the HP and NLSP equations,
using Hamiltonians that are given by a linear combination of the HP and NLSP Hamiltonians:
\begin{equation}
m_t = \D(\delta K^{(0)}/\delta\bar m) = \H(\delta K^{(1)}/\delta\bar m) 
\end{equation}
where
\begin{equation} \label{eq:K0K1def}
K^{(0)} = \cos \theta (-E^{(0)})+\sin \theta H^{(0)} ,
\quad
K^{(1)} = \cos \theta E^{(1)}  +\sin \theta\, H^{(1)} . 
\end{equation}
The unified equation \eqref{eq:mfamily} reveals an interesting symmetry 
between the HP and NLSP equations.  Indeed, 
 the HP and NLSP equations are respectively given by 
\begin{align} 
&m_t+(\re(Q)m)_x-i\im(Q)m=0, 
\quad
\theta =0, 
\label{eq:HP}
\\
&m_t-(\im(Q)m)_x-i \re(Q)m=0, 
\quad
\theta = \tfrac{\pi}{2}.
\label{eq:NLSP}
\end{align} 
Thus, these two equations are related by a phase rotation of $Q$ by the angle $\theta=\tfrac{\pi}{2}$. 

The bi-Hamiltonian formulations of the 
HP equation, NLSP equation, and the generalized equation \eqref{eq:mfamily}
exhibit the same symmetry. 
If we write 
\begin{equation} \label{eq:P}
P=(u-u_x)\bar m
\end{equation} 
then we have the relations
\begin{equation*}
\re(P) = \re(u\bar{m}) -\tfrac{1}{2}(\re(Q))_x,
\quad
\im(P) = -\im(u_x\bar{m}) -\tfrac{1}{2}(\im(Q))_x ,
\end{equation*}
where $\re(u\bar{m})$ is the density for $H^{(0)}$,
and $\im(u_x\bar{m})$ is the density for $E^{(0)}$. 
Likewise, for the quantity $QP$ we obtain
\begin{align*}
&\re(QP) = \re(Q)\re(u\bar m)-\im(Q) \im(\bar u_x m)+\tfrac14 (\im^2(Q) -\re^2(Q))_x\\
&\im(QP) = \re(Q)\im(\bar u_x m)+\im(Q) \re(\bar u m) -\tfrac{1}{2}(\im(Q)\re(Q))_x ,
\end{align*}
where $\tfrac14(\re(Q)\re(u\bar m)-\im(Q) \im(\bar u_x m))$ is the density for $H^{(1)}$ (see \eqref{eq:H1}) and $\tfrac14(\re(Q)\im(\bar u_x m)+\im(Q) \re(\bar u m))$ for $E^{(1)}$ (see \eqref{eq:E1}).  
These relations show that the densities (modulo irrelevant boundary terms) for $K^{(0)}$ and $K^{(1)}$ are given by 
$\im(e^{\inum \theta} P)$
and 
$\tfrac14 \im(e^{\inum \theta}QP)$,
respectively. 
Hence, 
we obtain 
\begin{equation}\label{eq:K0K1}
K^{(0)} = \int_{\RR} \im(e^{\inum \theta}  P) \,dx,
\quad
K^{(1)} = \int_{\RR} \tfrac14 \im(e^{\inum \theta}  QP) \,dx. 
\end{equation}
In particular, from these expressions for the unified Hamiltonians,
we see that the bi-Hamiltonian structures of the HP and NLSP equations 
are related by the phase rotation by $\pi/2$. 

In section~\ref{sec:peakons} we introduce a \textit{sector of conservative peakons} for \autoref{eq:mfamily}, concentrating mostly on the 
cases of $\theta=0$ (HP) and $\theta=\tfrac{\pi}{2}$ (NLSP).

\section{Conservative peakons} 
\label{sec:peakons}

The peakon Ansatz \cite{ch} 
\begin{equation*}  
u=\sum_{j=1}^Nm_je^{-|x-x_j|}
\end{equation*} 
was originally designed for real $m_j, x_j$.  
For the HP and NLSP equations \eqref{eq:HP} and \eqref{eq:NLSP}, 
the coefficients $m_j $ are complex and $x_j$ are real, 
resulting in $m=u-u_{xx}$ being a complex discrete measure
\begin{equation} \label{eq:peakon measure}
m=2\sum_{j=1}^Nm_j\delta_{x_j}.
\end{equation} 
Thus both equations \eqref{eq:HP} and \eqref{eq:NLSP}, 
and more generally \eqref{eq:mfamily},  
must be viewed as distribution equations.  
To this end the products $\im(Q)m$ and $ \re(Q)m$ need to be defined, 
and accordingly $Qm$ needs to be defined.  
By analyzing  the distributional Lax pair in a similar way to \cite{chang-szmigielski-m1CHlong}, 
we can show that the choice consistent with Lax integrability is to take 
\begin{equation} \label{eq:Qm} 
Qm=\avg{Q} m \stackrel{\text{def}}{=}2\sum_{j=1}^N\avg{Q}(x_j)m_j\delta_{x_j}, 
\end{equation} 
where $\avg{Q}(x_j)$ denotes the arithmetic average of the right and left hand limits at $x_j$.  

\begin{remark}
Many previous investigations of peakon equations, 
particularly on global existence and wave breaking for the mCH equation \cite{gui2013wave}, 
have defined distribution products differently by using a weak (integral) formulation of the peakon equation. 
The same approach was taken in \cite{anco} to derive single peakon weak solutions and peakon breather weak solutions of the HP and NLSP equations, 
but as pointed out in that paper, 
the HP and NLSP equations do not appear to have a weak formulation that allows multi-peakon solutions to be derived. 
Indeed, the choice of defining distribution products used here \eqref{eq:Qm} 
appears to be the only way to obtain multi-peakon solutions for these two equations, 
as well as for the general family \eqref{eq:mfamily}. 
As a consequence, the conservative single peakon and peakon breather solutions that will be obtained later in this paper differ from the single peakon weak solutions and peakon breather weak solutions presented in \cite{anco}. 
Most importantly, conservative $N$-peakon solutions will be derived for all $N\geq 1$. 
\end{remark}

Since $m_j$s are complex, we will use polar co-ordinates: 
\begin{equation*}
m_j=\abs{m_j} e^{i\omega_j}
\end{equation*}

Using these definitions, we obtain the following systems of ODEs 
from the peakon equations \eqref{eq:HP} and \eqref{eq:NLSP}. 

\begin{proposition} 
For the peakon Ansatz \eqref{eq:peakonansatz}, 
suppose the ill-defined product $Qm$ is regularized according to \eqref{eq:Qm}. 
Then  the HP equation \eqref{eq:HP}  reduces to
\begin{equation}\label{eq:pHP}
\dot x_j=\avg{\re Q}(x_j), 
\quad 
\dot \omega_j=\avg{\im Q}(x_j), 
\quad 
\frac{d|m_j|}{dt}=0,
\quad
j=1,\ldots,N . 
\end{equation}
Likewise, the NLSP equation \eqref{eq:NLSP} reduces to 
\begin{equation}\label{eq:pNLSP}
\dot x_j=-\avg{\im Q}(x_j), 
\quad 
\dot \omega_j=\avg{\re Q}(x_j), 
\quad 
\frac{d|m_j|}{dt}=0,
\quad
j=1,\ldots,N , 
\end{equation}
while in the general case of \autoref{eq:mfamily} the peakon ODEs read: 
\begin{equation}\label{eq:pmfamily}
\dot x_j=\avg{\re (e^{\inum \theta}Q)}(x_j), 
\quad 
\dot \omega_j=\avg{\im (e^{\inum \theta}Q)}(x_j), 
\quad 
\frac{d|m_j|}{dt}=0,
\quad
j=1,\ldots,N , \quad 0\leq \theta<\pi. 
\end{equation}
\end{proposition} 

It is easy to see that the vector fields on the right hand sides of 
equations \eqref{eq:pHP} and \eqref{eq:pNLSP} are orthogonal.  
The following conclusion about the geometry 
of solution curves of peakon ODEs is straightforward.  
\begin{corollary} \label{lem:HP-NLSP duality} 
The family of solution curves to the ODE system \eqref{eq:pHP} 
is orthogonal to the family of solution curves to the ODE system \eqref{eq:pNLSP}. In general, the family of solution curves to the ODE system \eqref{eq:pHP} is at the angle $\theta$ to the family of solution curves to the ODE system \eqref{eq:pmfamily}. 
\end{corollary}

We can write these ODE systems in a simpler form 
in terms of a complex variable 
\begin{equation} \label{eq:Xj}
X_j=x_j+\inum \omega_j
\end{equation}
which combines the positions and phases. 
\begin{lemma} 
The ODE systems \eqref{eq:pHP} and \eqref{eq:pNLSP}
 can be expressed in the complex-variable form 
\begin{align} 
& \dot X_j=\avg{Q}(x_j), 
\label{eq:pHPcplx}
\\
& \dot X_j=i\avg{Q}(x_j),  
\label{eq:pNLSPcplx}
\end{align}
and similarly for system \eqref{eq:pmfamily},
\begin{equation} \label{eq:pmfamcplx} 
\dot X_j=e^{\inum \theta}\avg{ Q}(x_j)  
\end{equation}
holds.
 \end{lemma}

\begin{remark} 
We recall \cite{chang-hu-szmigielski} that the peakons 
for the two-component modified Camassa-Holm (2mCH) equation 
satisfy an identically looking ODE system $\dot x_j=\avg{Q}(x_j)$, 
but with an important difference that $Q(x)=(u-u_x)(v+v_x)$,
where $(u,v)$ are the two (real) components. 
One natural reduction of the 2mCH equation is the modified Camassa-Holm (1mCH) equation obtained by putting $v=u$ \cite{chang-szmigielski-m1CHlong}.  
In a way the present paper is about the reduction $v=\bar u$.  
However, one needs to keep in mind that the work in \cite{chang-hu-szmigielski} is restricted to the real case, 
so the results of that paper do not apply in any direct way to the present situation.  
Nevertheless, for reasons that are not fully understood at this moment, 
the solution to the inverse problem associated with \eqref{eq:HP} or \eqref{eq:NLSP} 
turns out to have more similarities with the inverse problem for 1mCH peakons studied in \cite{chang-szmigielski-m1CHlong} 
rather than with the one for the 2mCH peakons in \cite{chang-hu-szmigielski}.   
\end{remark} 

\subsection{Poisson bracket}\label{sec:peakonPB}
We will now introduce a Poisson structure that will allow 
systems \autoref{eq:pHPcplx} and \autoref{eq:pNLSPcplx} to arise as Hamilton's equations.  
We observe that the vector field in equations \eqref{eq:pHP} and \eqref{eq:pNLSP}  
is not Lipschitz in the whole space $\RR^{2N}$ of $(x_j, \omega_j)$s.  
To remedy this, 
we will have to avoid the hyperplanes $x_i=x_j, i\neq j$, 
for example, by restricting our attention to the region of positions 
 where the ordering $x_1<x_2<\cdots<x_N$ holds.  Let us then 
 denote that region
 \begin{equation*} 
 \mathcal{P}=\{\mathbf x\in \RR^N: x_1<x_2<\cdots<x_N\} 
 \end{equation*} 
and, subsequently, define the pertinent phase space as follows.  
\begin{definition} 
\begin{equation} 
\mathcal{M}=\mathcal{P}\times T^N
\end{equation}
where $T^N$ is the $N$-dimensional torus of angles $\omega_1, \omega_2, \cdots, \omega_N$.  
\end{definition} 
Locally, it is convenient to think of a point $\xi=(x_1,x_2, \cdots, x_N, \omega_1, \omega_2, \cdots, \omega_N)\in \mathcal{M}$ as a complex 
vector $X=(X_1,X_2, \cdots, X_N)$ where $X_j$ was introduced in \autoref{eq:Xj}.  
Consequently, any function $f(\mathbf{\xi})\in \mathcal{C}^{\infty}(\mathcal{M})$  can be viewed as a smooth function $f$ of 
$X$ and its complex conjugate $\bar X$, namely $f=f(\mathbf{\xi})=f(X,\bar X)$.  
\begin{proposition} \label{lem:PB1}
The bracket 
\begin{equation} \label{eq:PB1}
\{X_j, X_k\}=\sgn(j-k), \quad \{X_j, \bar X_k\}=0, \quad \{\bar X_j, \bar X_k\}=\sgn(j-k), 
\end{equation}
defines a Poisson structure on $\mathcal{C}^{\infty}(\mathcal{M})$. 
\end{proposition} 
\begin{proof} 
It suffices to observe that \autoref{eq:PB1} is equivalent 
to 
\begin{equation}\label{eq:PB1R} 
\{x_j, x_k\}=\tfrac12 \sgn(j-k), \quad \{\omega_j, \omega_k\}=-\tfrac12 \sgn(j-k), \quad \{\omega_j, x_k\}=0. 
\end{equation}
This set of brackets defines a skew symmetric matrix $\Omega_{ab}$ 
on $\RR ^{2N}$ with block form 
\begin{equation*} 
\Omega=\begin{pmatrix} [\frac 12 \sgn(j-k)] &0\\ 0& -[\frac 12 \sgn(j-k)]\end{pmatrix}, 
\end{equation*}
each block having dimension $N\times N$.  
Then upon setting 
\begin{equation}\label{eq:PBfg}
\{f,g\}(\xi)=\sum_{a, b=1}^{2N} \Omega_{ab}\frac{\partial f}{\partial \xi_a}\frac{\partial g}{\partial \xi _b}
\end{equation}
we obtain the desired Poisson structure on $\mathcal{C}^{\infty}(\mathcal{M})$, 
since the skew symmetric matrix $\Omega$ is $\xi$-independent 
and thus the bracket \eqref{eq:PBfg} automatically satisfies the Jacobi identity.  
\end{proof} 
\begin{remark} 
Since $\Omega$ is full rank the Poisson bracket given by \autoref{eq:PB1}
equips $\mathcal{M}$ with a symplectic structure.  
\end{remark} 
Before we prove the main statement of this section we 
need to express $H^{(0)}$ and $E^{(0)}$ (see \autoref{eq:H0} and \autoref{eq:E0}) 
in terms of coordinates on $\mathcal{M}$.  
The detailed computations are provided in \autoref{sec:AppPB} (see also 
\autoref{lem:M1vH1} for a 
spectral interpretation of both quantities). 
 
\begin{lemma} \label{lem:H0E0} 
Let $u$ be given by the peakon Ansatz \eqref{eq:peakonansatz} and let 
the multiplication of the singular term $Qm$ be defined by \eqref{eq:Qm}.  
Then 
\begin{align} \label{eq:H0M} 
H^{(0)}\big|_{\mathcal{M}}&=4 \re\big(\sum_{k<l} \abs{m_k}\abs{m_l}
e^{X_k-X_l}\big)+2\sum_l \abs{m_l}^2, \\
E^{(0)}\big|_{\mathcal{M}}&=-4\im\big(\sum_{k<l} \abs{m_k}\abs{m_l}
e^{X_k-X_l}\big). \label{eq:E0M}
\end{align} 
\end{lemma}

\begin{theorem} \label{thm:XjHam}
\autoref{eq:pHP} and \autoref{eq:pNLSP} 
are Hamilton's equations of motion with respect to the Poisson 
structure given by \autoref{eq:PB1R} and Hamiltonians $H^{(0)}$ and 
$E^{(0)}$ respectively.  In terms of the complex variable $X $ we 
have 
\begin{equation}
\dot X_j=\{X_j, H^{(0)}\}  
\end{equation} 
for the HP peakon flow \eqref{eq:pHP}, 
\begin{equation}
\dot X_j=\{X_j, E^{(0)}\}  
\end{equation} 
for the NLSP peakon flow \eqref{eq:pNLSP}, and 
\begin{equation}
\dot X_j=\{X_j, K^{(0)}\}  
\end{equation} 
for the general peakon flow \eqref{eq:pmfamily}. 
\end{theorem} 
\begin{proof} 
We will first compute $\{X_j, H^{(0)}\}$ using \autoref{eq:H0M}; 
for convenience we abbreviate \textit{ c.c.} to mean the complex conjugate.  
We have 
\begin{equation*} \begin{split} 
\{X_j, H^{(0)}\}=&2 \{X_j, \sum_{k<l}\abs{m_k}\abs{m_l} e^{X_k-X_l}+ \text{c.c.}\}\stackrel{\autoref{eq:PB1}}{=}
2\sum_{k<l} \abs{m_k}\abs{m_l} e^{X_k-X_l} \{X_j, X_k-X_l\}=\\&2\sum_{k<l} \abs{m_k}\abs{m_l} e^{X_k-X_l}\big(\sgn(j-k)-\sgn(j-l) \big)=4 \sum_{k<j<l} 
\abs{m_k}\abs{m_l} e^{X_k-X_l} +\\&2\abs{m_j}\big(\sum_{k<j}\abs{m_k}e^{X_k-X_j}+\sum_{j<k} \abs{m_k} e^{X_j-X_k}\big)\stackrel{\autoref{lem:Qavg}}{=} \avg{Q}(x_j). 
\end{split} 
\end{equation*} 
Likewise, 
\begin{equation*} \begin{split} 
\{X_j, E^{(0)}\}=&2\inum \{X_j, \sum_{k<l}\abs{m_k}\abs{m_l} e^{X_k-X_l}-\text{c.c.}\}\stackrel{\autoref{eq:PB1}}{=}
2\inum\sum_{k<l} \abs{m_k}\abs{m_l} e^{X_k-X_l} \{X_j, X_k-X_l\}=\\&2\inum\sum_{k<l} \abs{m_k}\abs{m_l} e^{X_k-X_l}\big(\sgn(j-k)-\sgn(j-l) \big)=4\inum \sum_{k<j<l} 
\abs{m_k}\abs{m_l} e^{X_k-X_l} +\\&2\inum\abs{m_j}\big(\sum_{k<j}\abs{m_k}e^{X_k-X_j}+\sum_{j<k} \abs{m_k} e^{X_j-X_k}\big)\stackrel{\autoref{lem:Qavg}}{=} \inum\avg{Q}(x_j). 
\end{split} 
\end{equation*} 
Finally, the general case can be verified by using the above results and 
\eqref{eq:K0K1def}.  
\end{proof} 
\begin{remark} 
In addition to the Poisson bracket \eqref{eq:PB1},
there is a second Poisson structure on $\mathcal{C}^{\infty}(\mathcal{M})$
defined by another bracket 
\begin{equation} \label{eq:PB2}
\{X_j, X_k\}_{\frac \pi 2}=\inum \sgn(j-k), \quad \{X_j, \bar X_k\}_{\frac \pi 2}=0, \quad \{\bar X_j, \bar X_k\}_{\frac \pi 2}=-\inum \sgn(j-k), 
\end{equation}
or, equivalently, 
\begin{equation}\label{eq:PB2R} 
\{x_j, x_k\}_{\frac \pi 2}=0, \quad \{\omega_j, \omega_k\}_{\frac \pi 2}=0, \quad \{\omega_j, x_k\}_{\frac \pi 2}=\tfrac 12 \sgn(j-k). 
\end{equation}
The rationale for the subscript $\frac \pi 2$ will be explained below, but 
for now 
we note that the skew symmetric matrix $\Omega$ takes the form: 
\begin{equation*} 
\Omega_{\frac \pi 2}=\begin{pmatrix} 0&[\frac12\sgn(j-k)]\\
[\frac12 \sgn(j-k)]&0 \end{pmatrix},  
\end{equation*}
and both of the peakon equations \eqref{eq:pHP} and \eqref{eq:pNLSP} 
remain Hamiltonian, although with swapped Hamiltonians
\begin{equation} 
\dot X_j=\{X_j, H^{(0)}\}_{\frac \pi 2}
\end{equation} 
for the peakon NLSP equation \eqref{eq:pNLSPcplx} 
and 
\begin{equation}
\dot X_j=\{X_j, -E^{(0)}\}_{\frac \pi 2}
\end{equation}
for the peakon HP equation \eqref{eq:pHPcplx}.  
\end{remark}

The second bracket appears to be more natural one, since $-E^{(0)}$ is the Hamiltonian for HP and since this bracket arises from reduction of the Hamiltonian structure given by $\D$, though the reduction is slightly singular.  This point will be elaborated on elsewhere.  However, 
there is another, perhaps more unifying, point of view that we would like to mention here.  
To this end we define a $\theta$-dependent 
Poisson structure
\begin{definition} 
\begin{equation} \label{eq:PBtheta}
\{X_j, X_k\}_{\theta}=e^{\inum \theta} \sgn(j-k), \quad \{X_j, \bar X_k\}_{\theta}=0, \quad \{\bar X_j, \bar X_k\}_{\theta}=e^{-\inum \theta}\sgn(j-k), 
\end{equation}
or, equivalently, 
\begin{equation}\label{eq:PBthetaR} 
\{x_j, x_k\}_{\theta}=\tfrac{\cos \theta}{2} \sgn(j-k), \quad \{\omega_j, \omega_k\}_{\theta}=-\tfrac{\cos \theta}{2} \sgn(j-k), \quad \{\omega_j, x_k\}_{\theta}=\tfrac{\sin  \theta}{2} \sgn(j-k). 
\end{equation}
\end{definition} 
The skew symmetric matrix $\Omega$ now takes the form
\begin{equation*} 
\Omega_{\theta} =\begin{pmatrix} [\tfrac{\cos \theta}{2}\sgn(j-k)] &[\tfrac{\sin \theta}{2} \sgn(j-k)]\\
[\tfrac{\sin \theta}{2}\sgn(j-k)]&- [\tfrac{\cos \theta}{2}\sgn(j-k)] \end{pmatrix},  
\end{equation*}
which clearly combines both previous cases.  
More importantly, the following lemma holds, the proof of which is 
just a simple modification of the proof of \autoref{thm:XjHam}.  
\begin{lemma} \label{lem:XjHamfam}
The $\theta$ family of equations \eqref{eq:pmfamcplx} is Hamiltonian with respect to the 
Poisson bracket \eqref{eq:PBtheta} with a fixed 
Hamiltonian $H^{(0)}$, that is \autoref{eq:pmfamcplx} can be written
\begin{equation} 
\dot X_j=\{X_j, H^{(0)}\}_{\theta}, \qquad 0\leq \theta< \pi.  
\end{equation} 
\end{lemma} 
We will conclude this subsection by stating an easy corollary 
focusing again on special cases of HP and NLSP equations, followed 
by a theorem about the norm preservation for  
 the $\theta$ family.   
\begin{corollary} 
In the original variables $(m_j, x_j)$, and written in the notation consistent with \eqref{eq:PBtheta}, 
the Poisson brackets \eqref{eq:PB1} and \eqref{eq:PB2}  are given by, respectively, 
\begin{align*}
& 
\{m_j,m_k\}_0=\tfrac 12 \sgn(j-k) m_j m_k,
\quad
\{m_j,\bar m_k\}_0=-\tfrac 12 \sgn(j-k) m_j \bar m_k,
\\
&
\{x_j,x_k\}_0=\tfrac 12 \sgn(j-k),
\\
& 
\{x_j,m_k\}_0= 0,\\
&\\
\text{and } \qquad &\\
&\\
& \{m_j,m_k\}_{\frac \pi 2}=0, \quad \quad \{m_j, \bar m_k\}_{\frac \pi 2}=0, 
\\
&
\{x_j,x_k\}_{\frac \pi 2}=0,
\\
& 
\{x_j,m_k\}_{\frac \pi 2}= \tfrac{\inum}{2}\, \sgn(j-k) m_j . 
\end{align*}
\end{corollary} 
Both Poisson structures can be derived from the first Hamiltonian structure of the NLSP and HP equations by a (singular) reduction process.  
This topic will be  taken up in another publication.  

We recall, as shown in section~\ref{sec:HamNLS},  
that $H^{(0)}$ is conserved for both the HP and NLSP equations; this was then further 
amplified in the peakon sector for all equations in the $\theta$ family of equations (\autoref{lem:XjHamfam}).  
We stress that, at least in the peakon sector, all equations 
in the $\theta$ family share the same Hamiltonian, but 
their Hamiltonian structure deforms.  
Since $||u||^2_{H^1}=H^{(0)}$ is the square of the Sobolev norm, 
we have the following theorem which justifies the name `` conservative peakons''.  We emphasize that this theorem is valid not only 
for the HP and NLSP peakons but, thanks to \autoref{lem:XjHamfam}, 
for the whole peakon $\theta$  family \eqref{eq:pmfamcplx}.  
\begin{theorem} \label{thm:H1conserve_NLSP_HP} 
Let $u$ be given by the peakon Ansatz \eqref{eq:peakonansatz} 
and let the singular term $Qm$ be regularized by \eqref{eq:Qm}.  
Then for any $0\leq \theta<\pi$:
\begin{equation*} 
\frac{d}{dt} ||u||_{H^1}=0.  
\end{equation*} 
\end{theorem}

\section{HP equation; the spectral theory} \label{sec:HPspectral}
For the reminder of this work we will concentrate 
mostly on the HP case, and to some extent on the NLSP case, leaving more 
in-depth analysis of the $\theta$ family for future investigations. 

As was indicated in \autoref{rem:gl2} the Lax pair can be 
modified by a multiple of identity.  This  is effectively changing 
what appeared to be an $\mathfrak{sl}(2, \C)$ theory to a $\mathfrak{gl}(
2, C)$ theory.  
We take the Lax pair for the HP equation \autoref{eq:HP} to be 
(compare with \eqref{eq:UV}, \eqref{eq:Q}, \eqref{eq:J})
\begin{equation} \label{eq:HPLax}
\Psi_x=U \Psi, \quad  \Psi _t =V \Psi, \quad  \Psi=\begin{bmatrix} \Psi_1\\\Psi_2 \end{bmatrix} ,
\end{equation} 
where
\begin{align*} 
&U=\tfrac{1}{2}\begin{bmatrix} -1 &\lambda m\\ -\lambda \bar{m}& 1 \end{bmatrix},\\
\\ 
&V=\tfrac{1}{2}\begin{bmatrix} 4\lambda^{-2} + Q & -2\lambda^{-1} (u-u_x)-\lambda m \re(Q)\\
2\lambda^{-1}(\bar u+\bar u_x)+\lambda \bar m \re(Q)& -Q \end{bmatrix},
\end{align*} 
with $Q$ given by expression \eqref{eq:Q}.  This choice 
of $V$ is compatible, as opposed to $V$ in \autoref{eq:UV}, with the asymptotic behaviour 
$\Psi=\begin{bmatrix} 0\\ e^x \end{bmatrix}$ as $x\rightarrow -\infty$.  This type of asymptotic adjustment is present in all peakon equations known to 
us (e.g. \cite{bss-moment}, \cite{ls-cubicstring}). 
Performing on \eqref{eq:HPLax} a $GL(2,\C)$ gauge transformation $$\Phi=\textrm{diag}(\lambda^{-1}e^{\frac x2}, e^{-\frac x2}) \Psi$$ 
yields a simpler $x$-equation
\begin{equation}\label{eq:xLax}
\Phi_x=\begin{bmatrix}0 & h\\
-z g& 0 \end{bmatrix} \Phi, \qquad    g=\sum_{j=1} ^Ng_j \delta_{x_j}, \qquad h=\sum_{j=1} ^Nh_j \delta_{x_j}, 
\end{equation} 
where $g_j=\bar m_j e^{-x_j}, \, h_j =m_j e^{x_j}, \, z=\lambda^2 $. 
For future use we note, using the complex-variable notation \eqref{eq:Xj}, that 
\begin{equation} \label{eq:ghX}
g_j=\abs{m_j}e^{- X_j} , \qquad h_j =\abs{m_j} e^{X_j}, 
\end{equation} 
hence $g_jh_j=|m_j|^2$.  

We can impose the boundary conditions $\Phi_1(-\infty)=0$ and $\Phi_2(+\infty)=0$ without violating the compatibility of the Lax pair \eqref{eq:HPLax}.  
The argument in support of that is similar to other peakon cases, 
most notably to the modified CH equation \cite{chang-szmigielski-m1CHlong}, 
so we skip it in this paper.  
However, to make the boundary value problem 
\begin{equation}\label{eq:xLaxBVP}
\Phi_x=\begin{bmatrix}0 &h\\
-z g& 0 \end{bmatrix} \Phi, \qquad \Phi_1(-\infty)=\Phi_2(+\infty)=0, 
\end{equation}
well posed,  
we need to define the multiplication of the measures $h$ and $g$ by $\Phi$ on their singular support, namely at the points $x_j$.  
It can be shown in a way similar to what was done in \cite{chang-szmigielski-m1CHlong} that 
if we require that $\Phi$ be left continuous and define $\Phi_{a}\delta_{x_j}=\Phi_a(x_j)\delta_{x_j}, a=1,2$,  
then this choice makes the Lax pair \eqref{eq:HPLax} well defined as a distributional Lax pair, 
and the compatibility condition of the $x$ and $t$ components of the Lax pair indeed 
implies the peakon HP equation \eqref{eq:pHP}.  

The solution $\Phi$ is a piecewise constant function in $x$ which, for convenience, we can normalize by setting $\Phi_2(-\infty)=1$.  The distributional boundary value problem \eqref{eq:xLaxBVP}, whenever $m$ is a discrete measure, is equivalent to a finite difference equation. 
\begin{lemma}\label{lem:forwardR}
Let $q_k=\Phi_1(x_{k}+), \, ~p_k=~\Phi_2(x_{k}+),$ then the finite-difference form of the 
boundary value problem is given by
\begin{equation} \label{eq:dstring}
\begin{gathered}
\begin{aligned}
q_{k}-q_{k-1}&=h_kp_{k-1}, & 1\leq k\leq N, \\
p_{k}-p_{k-1}&=-z g_kq_{k-1},& 1\leq k\leq N,\\
q_0=0, \quad  p_0=1&, \quad p_{N}(z)=0.  &  
\end{aligned}
\end{gathered}
\end{equation} 
\end{lemma} 
An easy proof by induction leads to the following result  
for the associated initial value problem.  
\begin{lemma} \label{cor:pq-degrees} 
Consider the initial value problem
\begin{equation} \label{eq:dstringIVP}
\begin{gathered}
\begin{aligned}
q_{k}-q_{k-1}&=h_kp_{k-1}, & 1\leq k\leq N, \\
p_{k}-p_{k-1}&=-z g_kq_{k-1},& 1\leq k\leq N,\\
q_0=0, &\quad  p_0=1. & 
\end{aligned}
\end{gathered}
\end{equation} 
Then $q_{k}(z)$  is a polynomial of degree $\lfloor\frac{k-1}{2}\rfloor$ in $z$, and $p_{k}(z)$ is a polynomial of degree $\lfloor\frac{k}{2}\rfloor$, respectively. 
\end{lemma} 
We remark that the finite-difference form of the boundary value problem \eqref{eq:dstringIVP} admits a simple matrix representation
\begin{equation}\label{transition}
\begin{bmatrix}
  q_{k}\\
  p_{k}
\end{bmatrix}
=T_k\begin{bmatrix}
  q_{k-1}\\
  p_{k-1}
\end{bmatrix}, \qquad\qquad
T_k=\begin{bmatrix}
  1& h_k\\
  -z g_k&1
\end{bmatrix}, 
\end{equation}
and observe that in view of \eqref{eq:ghX}
\begin{equation}\label{eq:det Tk} 
\det T_k=1+\abs{m_k}^2 z. 
\end{equation}

\begin{definition} 
A complex number z is an \textit{eigenvalue} of the boundary value problem 
\eqref{eq:dstring} if there exists a solution $\{q_k, p_k\}$ to \eqref{eq:dstringIVP} for 
which $p_{N}(z)=0$.   The set of all eigenvalues is the \textit{spectrum}  of the boundary value problem \eqref{eq:dstring}.  
\end{definition} 

\begin{remark} 
Clearly, $z=0$ is not an eigenvalue. 
\end{remark}

To encode the spectral data we introduce 
the \textit{Weyl function} 
\begin{equation}\label{eq:defWeyl}
W(z)=\frac{q_{N}(z)}{p_{N}(z)}.  
\end{equation} 
If the spectrum of the boundary problem \eqref{eq:dstring} is simple, $W(z)$ can be written as
\begin{equation}\label{eq:simpleW}
W(z)=c+\sum_{j=1}^{\lfloor \frac N2 \rfloor} \frac{b_j}{\zeta_j-z}.  \end{equation}
\begin{remark} In contrast to the situation for the 1mCH equation in \cite{chang-szmigielski-m1CHlong} we no longer expect in general the spectrum to be either simple or real.  
\end{remark} 

Regardless of the nature of the spectrum we easily obtain the following 
result by  examining the $t$ part of the Lax pair \eqref{eq:HPLax} in the region $x>x_{N}$.  
\begin{lemma}\label{lem:t-evolution of qp} 
Let $\{q_k, p_k\}$ satisfy the system of difference equations \eqref{eq:dstringIVP}.  
Then 
\begin{equation}\label{eq:tderqp}
\dot q_{N}=\frac{2}{z}q_{N}-\frac{2L}{z}\,p_{N}, \qquad \dot p_{N}=0,
\end{equation}
where $L=\sum_{j=1}^{N}h_j$. 
Thus $p_{N}(z)$ is independent of time and, in particular, its zeros, i.e. the spectrum, are time invariant.
Moreover, 
\begin{equation}\label{eq:dotW}
\dot W=\frac{2}{z}W-\frac{2L}{z}.  
\end{equation} 
\end{lemma} 
If the spectrum is simple we have further simplification of the time 
evolution.  
\begin{corollary} 
Suppose $p_N(z)$ has simple roots.  
Then the data in the Weyl function \autoref{eq:simpleW} has the time evolution
\begin{equation} \label{eq:tflowSD}
\dot c=0, \quad \dot \zeta_j=0, \quad \dot b_j=\frac{2}{\zeta_j} b_j. 
\end{equation} 
\end{corollary}

Let us recall a notation introduced in 
\cite{chang-szmigielski-m1CHlong} to present in a compact form expressions  
appearing in the solution to the inverse problem; these expressions call for choices of $j$-element index sets $I$ and~$J$
from the set $[k] = \{ 1,2,\dots,k \}$.
Henceforth we will use the notation
$\binom{[k]}{j}$ for the set of all $j$-element subsets of $[k]$, listed in increasing order; for example $I\in \binom{[k]}{j}$ means that 
$I=\{i_1, i_2,\dots, i_j\}$ for some increasing sequence $i_1 < i_2 < \dots < i_j\leq ~k$. 
Furthermore, given the multi-index $I$ and a vector $\mathbf{g}=(g_1,g_2, \cdots, g_k)$ we will abbreviate $\mathbf{g}_I=g_{i_1}g_{i_2}\dots g_{i_j}$ etc.
 
\begin{definition}\label{def:bigIndi} 
Let $I,J \in \binom{[k]}{j}$, or $I\in \binom{[k]}{j+1},J \in \binom{[k]}{j}$.  
Then  $I, J$ are said to be \emph{interlacing} if 
\begin{equation*}\label{eq:interlacing}
i_{1} <j_{1} < i_{2} < j_{2} < \dotsb < i_{j} <j_{j}
\end{equation*}
or, 
\begin{equation*}
i_{1} <j_{1} < i_{2} < j_{2} < \dotsb < i_{j} <j_{j}<i_{j+1}, 
\end{equation*}
in the latter case.  
We abbreviate this condition as $I < J$ in either case, and, furthermore, 
use this same notation for $I\in \binom{[k]}{1}, J \in \binom{[k]}{0}$.
\end{definition} 
By a straightforward computation of the coefficients of 
$p_N=1-M_1z+\cdots+\cdots M_j (-z)^j+\cdots $ (see Corollary 2.7 in \cite{chang-szmigielski-m1CHlong}) 
we obtain the following description of constants of motion.  
\begin{lemma} \label{lem:constants} 
The quantities 
\begin{equation*} 
M_j=\sum_{\substack{I,J \in \binom{[N]}{j}\\ I<J}} h_I g_J, \qquad 1\leq j\leq \lfloor\tfrac{N}{2}\rfloor
\end{equation*}
comprise a set of $\lfloor\frac{N}{2}\rfloor$ constants of motion for the system \eqref{eq:pHP}.  
\end{lemma}
\begin{example}
Let us consider the case $N=4$.  
Then the constants of motion, written in terms of the complex variables $X_j$ (see \eqref{eq:Xj}), with positions $x_j$ satisfying $x_1<x_2<x_3<x_4$, are
\begin{align*}
M_1&=\abs{m_1 m_2}e^{X_1-X_2}+\abs{m_1 m_3}e^{X_1-X_3}+\abs{m_1m_4}e^{X_1-X_4}+\abs{m_2 m_3}e^{X_2-X_3}+\abs{m_2m_4}e^{X_2-X_4}+\\&\abs{m_3m_4}e^{X_3-X_4},\\
M_2&=\abs{m_1m_2m_3 m_4}e^{X_1-X_2+X_3-X_4}.
\end{align*}
\end{example}

We have the following, very preliminary, characterization of the spectrum.
\begin{lemma}\label{lem:spec_nonneg}
If all the angles $\omega_j$ in the parametrization given by \autoref{eq:Xj} satisfy 
\begin{equation}\label{eq:anglecondtion}
-\frac{\pi}{2N}\leq \omega_j \leq\frac{\pi}{2N},
\end{equation}
then the spectrum of the boundary value problem  \eqref{eq:dstring} is a finite subset of $$\big\{z| ~-\pi<\arg(z)<\pi\big\}, $$  
namely, there are no eigenvalues on the negative real axis. 
\end{lemma}
\begin{proof}
Suppose there exists a positive number $\zeta_0>0$ for which $-\zeta_0$ is an eigenvalue,
hence 
$$p_N(-\zeta_0)=0.$$
By Corollary (2.7) in \cite{chang-szmigielski-m1CHlong} we have
$$p_N(z)=1+\sum_{j=1}^{\lfloor\frac{N}{2}\rfloor}\Big(\sum_{\substack{I,J \in \binom{[N]}{j}\\ I<J}} h_I g_J
\, \Big)(-z)^j.  $$
Under condition \eqref{eq:anglecondtion}, 
and recalling the parametrization of $g_i$s and $h_j$s  (see \eqref{eq:ghX}), 
it is straightforward to see that the coefficients of $(-z)^j$ satisfy
$$
-\frac{\pi}{2}\leq \arg\Big(\sum_{\substack{I,J \in \binom{[k]}{j}\\ I<J}} h_I g_J\, \Big)\leq \frac \pi 2
$$
leading to 
$$\re\Big(\sum_{\substack{I,J \in \binom{[k]}{j}\\ I<J}} h_I g_J\, \Big)\geq0.$$
Therefore, we have 
$$\re\left(p_N(-\zeta_0)\right)=1+\sum_{j=1}^{\lfloor\frac{N}{2}\rfloor}\re\Big(\sum_{\substack{I,J \in \binom{[k]}{j}\\ I<J}} h_I g_J\, \Big)(\zeta_0)^j>0,$$
contradicting $p_N(-\zeta_0)=0$.
\end{proof}

With the additional assumptions on the angles $\omega_j$ in place, 
we can improve upon \eqref{eq:simpleW}.  
\begin{lemma} \label{lem:W} 
Let $W$ be the Weyl function \eqref{eq:defWeyl}. 
Suppose the spectrum of the boundary problem \eqref{eq:dstring} is simple, 
and $\omega_j$s satisfy condition \eqref{eq:anglecondtion}. 
Then $W(z)$ can be expressed as
\begin{equation}
W(z)=c+\sum_{j=1}^{\lfloor \frac N2 \rfloor} \frac{b_j}{\zeta_j-z}, \qquad  b_j\neq0, 
\end{equation}
where $c\neq0$ when $N$ is odd, and $c=0$ when $N$ is even.  
\end{lemma} 
\begin{proof}
The claim about $c$ can be verified by examining the degree of $p_N$ and $q_N$.

To prove the nonzero property of $b_j$ we suppose 
that there exists some $b_j=0$, which means that  for some $\zeta_0$
$$p_N(\zeta_0)=q_N(\zeta_0)=0.$$
By the recursive relation \eqref{eq:dstringIVP} we have
$$-q_{N-1}(\zeta_0)=h_{N}p_{N-1}(\zeta_0),\qquad p_{N-1}(\zeta_0)=\zeta_0 g_Nq_{N-1}(\zeta_0),$$
leading to 
$$\left(1+\zeta_0|m_N|^2\right)p_{N-1}(\zeta_0)=0. $$
Furthermore, since $\zeta_0$ can not be negative by \autoref{lem:spec_nonneg}, we obtain
$p_{N-1}(\zeta_0)=0. $
Then the second relation above, taking into account that $0$ is not an eigenvalue, implies
$q_{N-1}(\zeta_0)=0$.  
By implementing the above argument and using \eqref{eq:dstringIVP} recursively we eventually get 
$$
p_{1}(\zeta_0)=q_{1}(\zeta_0)=0,
$$ 
thus contradicting $p_1=1$ obtained from the first iteration of 
\eqref{eq:dstringIVP}.  Therefore the proof is completed.
\end{proof}

We finish this section by commenting about the connection between the constant of motion $M_1$
and the two conserved Hamiltonians $H^{(0)}$ and $E^{(0)}$ 
(see \eqref{eq:H0} and \eqref{eq:E0}).

\begin{lemma}\label{lem:M1vH1}
The Hamiltonians $H^{(0)}$ and $E^{(0)}$ are related to the constant of motion $M_1$
by
\begin{equation} 
H^{(0)}=2\sum_{j=1}^N|m_j|^2+4\re(M_1), 
\qquad 
E^{(0)} =-4\im(M_1) , 
\end{equation} 
where each $|m_j|$ is itself a constant of motion. 
\end{lemma} 
\begin{proof}
From \autoref{lem:constants}, 
$$
M_1=\sum_{1\leq j<k\leq N}|m_j| |m_k|e^{X_j-X_k}.
$$
Then the result follows immediately from \autoref{lem:H0} and \autoref{lem:E0}.  
\end{proof} 

\section{Inverse Problems}\label{sec:IP}
\subsection{First inverse problem: spectral data $\Rightarrow (g_j, h_j)$ }\label{sec:fIP}
We will formulate the inverse problem in the case where the spectrum is 
given by a collection of distinct complex 
numbers $\zeta_j, \, 1\leq j \leq  {\lfloor \frac N2 \rfloor}$, none of which 
lies on the negative real axis (see \autoref{lem:spec_nonneg} for 
one scenario ensuring the validity of the latter condition). 
\begin{definition} \label{def:ISP}
Given a rational function 
\begin{equation} \label{eq:Wc}
W(z)=c+\sum_{j=1}^{\lfloor \frac N2 \rfloor} \frac{b_j}{\zeta_j-z},  
\end{equation}
and a collection of distinct positive numbers $|m_j|$,  
find complex constants $g_j, h_j$,  for $1\leq j\leq N$, such 
that $g_jh_j=|m_j|^2$ and 
also such that the solution of the initial value problem
\begin{equation*} 
\begin{gathered}
\begin{aligned}
     q_{k}-q_{k-1}&=h_kp_{k-1}, & 1\leq k\leq N, \\
     p_{k}-p_{k-1}&=-z g_kq_{k-1},& 1\leq k\leq N,\\
      q_0=0, &\quad  p_0=1, &  
  \end{aligned}
\end{gathered}
\end{equation*} 
satisfies 
\begin{equation*}
W(z)=\frac{q_N(z)}{p_N(z)}.  
\end{equation*} 
\end{definition}
\begin{remark} 
The non degeneracy condition that the positive 
numbers $\abs{m_j}$ be distinct will be eventually relaxed; 
the condition simplifies the derivation of the inverse formulas. 
\end{remark}  
First we give a brief summary of main ideas behind the solution 
of the inverse problem stated in \autoref{def:ISP}.  
The main tool is a certain {\bf interpolation problem} (see \cite{chang-szmigielski-m1CHlong} for details).  In short, let us rewrite \eqref{transition} in terms of the Weyl function $W$, iterating down \eqref{transition} $k$ times starting with the highest index $N$:
\begin{equation} \label{eq:Papprox1}
\begin{bmatrix} W(z) \\ 1 \end{bmatrix} =T_N(z)T_{N-1}(z)  \dots T_{N-k+1}(z)\begin{bmatrix} \frac{q_{N-k}(z)}{p_{N}(z)}\\ \frac{p_{N-k}(z)}{p_{N}(z)} \end{bmatrix}. 
\end{equation}
Then by using the transpose of the matrix of cofactors (\textit{adjugate}) of each $T_j(z)$, and denoting \newline $\begin{bmatrix} 1 &- h_j\\zg_j &1 \end{bmatrix}\stackrel{\textrm {def}}{=}C_{N-j+1}(z)$,  
one can express equation \eqref{eq:Papprox1}  as 
\begin{equation*}
C_{k}(z)\dots C_1(z)\begin{bmatrix} W(z) \\ 1 \end{bmatrix} =\det (T_N(z))\det (T_{N-1}(z))  \dots \det(T_{N-k+1}(z))\begin{bmatrix} \frac{q_{N-k}(z)}{p_{N}(z)}\\ \frac{p_{N-j}(z)}{p_{N}(z)} \end{bmatrix}.  
\end{equation*}
Recalling that $\det T_j(z)=1+z\abs{m_j}^2$ and using our assumption that none of the roots of $p_{N}(z)$ lies on the negative real axis, we conclude
\begin{equation} \label{eq:Papprox2}
\left( C_{k}(z)\dots C_1(z)\begin{bmatrix} W(z) \\ 1 \end{bmatrix}\right)\Big |_{z=-\frac{1}{\abs{m_{N-i+1}}^2}}=0,  \qquad \textrm{ for any } 1\leq i\leq k.  
\end{equation}
\autoref{eq:Papprox2} can be interpreted as 
an interpolation problem.   
\begin{theorem} [\cite{chang-szmigielski-m1CHlong}]\label{thm:CJ interp}
Let the matrix of products of $C$s in equation \eqref{eq:Papprox2} be denoted by \newline $\begin{bmatrix} a_k(z)&b_k(z)\\c_k(z)&d_k(z) \end{bmatrix}
~\stackrel{def}{=} \hat S_k(z)$.  Then the polynomials $a_k(z), b_k(z),c_k(z),d_k(z)$ 
solve the following interpolation problem: 
\begin{subequations}\label{eq:Papprox}
\begin{align}
&a_k(-\frac{1}{\abs{m_{N-i+1}}^2})W(-\frac{1}{\abs{m_{N-i+1}}^2})+b_k(-\frac{1}{\abs{m_{N-i+1}}^2})
=0, \qquad 1\leq i\leq k, \\
&\deg a_k=\big\lfloor \frac{k}{2} \big \rfloor, \qquad \deg b_k=\big \lfloor \frac{k-1}{2} \big\rfloor, \qquad a_k(0)=1, \\
\notag\\
&c_k(-\frac{1}{\abs{m_{N-i+1}}^2})W(-\frac{1}{\abs{m_{N-i+1}}^2})+d_k(-\frac{1}{\abs{m_{N-i+1}}^2})
=0, \qquad 1\leq i\leq k,  \\
&\deg c_k=\big \lfloor \frac{k+1}{2}\big \rfloor, \qquad \deg d_k=\big \lfloor \frac{k}{2}\big \rfloor, \qquad c_k(0)=0, \quad d_k(0)=1.   
\end{align}
\end{subequations} 
\end{theorem} 
\begin{remark} 
The appearance of the label $N-j+1$ in the above 
formulation is fully explained in \cite{chang-szmigielski-m1CHlong}; roughly, this way of counting is typical of 
the right, rather than left, initial value problem.  This effectively results in 
the \textit{counting} of  the masses from right to left rather than from left to right. 
\end{remark} 
Now we outline our strategy for solving the 
inverse problem given by \autoref{def:ISP}: 
\begin{enumerate} 
\item given $W$ and $\{\abs{m_j}\}$ we solve 
the interpolation problem of \autoref{thm:CJ interp} for polynomials 
$a_k, b_k, c_k, d_k$, or equivalently the matrix $\hat S_k(z)$ (see 
the theorem above for the definition); 
\item using the relation between $\hat S_k(z)$ and 
the transition matrices $T_j$s which depend on $h_j$s and $g_j$s,  
we establish how the coefficients in the polynomials $a_k, b_k, c_k, d_k$
are built out of $h_j$s and $g_j$s, leading to formulas 
expressing $h_j$s and $g_j$s as ratios of certain coefficients of 
$a_k, b_k, c_k, d_k$.  

\end{enumerate} 

The algebraic solution to the interpolation problem stated in \autoref{thm:CJ interp}
was essentially given in \cite{chang-szmigielski-m1CHlong} with one important \textit{caveat}: the problem is now complex 
since both the spectrum and the residues $b_j$ in \eqref{eq:Wc} are complex.  Luckily, even though this affects the global existence when time is switched on, 
it nevertheless has no bearing on the algebraic formulation.  

We begin the presentation of formulas by introducing some additional notation.  
Thus we denote: 
 $[i,j]=\{i,i+1,\cdots,j\}, \,\,  \binom{[1,K]}{k}=\{J=\{j_1,j_2,\cdots,j_k\}|j_1<\cdots <j_k, j_i\in [1,K]\}$ and $i'=N-i+1$.  Moreover, given two sets of 
 vectors $\mathbf{x}=(x_1, x_2,\dots, x_N)$, $\mathbf{y}=(y_1, y_2, \cdots, y_N)$ and two ordered multi-index sets $I, J$  we define 
\begin{align*}
  \mathbf{x}_J&=\prod_{j\in J}x_j, &\Delta_J(\mathbf{x})&=\prod_{i<j\in J}(x_j-x_i), \\
  \Delta_{I,J}(\mathbf{x};\mathbf{y})&=\prod_{i\in I}\prod_{j\in J}(x_i-y_j), 
  &\Gamma_{I,J}(\mathbf{x};\mathbf{y})&=\prod_{i\in I}\prod_{j\in J}(x_i+y_j),  
  \end{align*}
  along with the convention
\begin{align*}
&\Delta_\emptyset(\mathbf{x})=\Delta_{\{i\}}(\mathbf{x})=\Delta_{\emptyset,J}(\mathbf{x};\mathbf{y})=\Delta_{I,\emptyset}(\mathbf{x};\mathbf{y})=\Gamma_{\emptyset,J}(\mathbf{x};\mathbf{y})=\Gamma_{I,\emptyset}(\mathbf{x};\mathbf{y})=1,&\\
&\binom{[1,K]}{0}=1;\qquad\qquad \binom{[1,K]}{k}=0,\ \  k>K.&
\end{align*}

Building in an essential way on work \cite{chang-szmigielski-m1CHlong} we can now give an algebraic solution to the inverse problem 
stated in \autoref{def:ISP}, postponing more delicate issues like 
global existence to future studies.  
\begin{theorem}
Given the vector $\mathbf{e}=(e_1,e_2,\cdots, e_N), e_j=\frac{1}{|m_{j'}|^2}$ as well as the vector of residues of the Weyl function $\mathbf{b}=(b_1,b_2, 
\cdots, b_{\lfloor \frac N2 \rfloor})$, let 
\begin{align}\label{eq:D}
\D_k^{(l,p)}=\left\{
\begin{array}{ll}
 \Delta_{[1,k]}(\mathbf{e})
\sum\limits_{\substack{J \in \binom{[\lfloor \frac N2 \rfloor]}{j}}}  \frac{\Delta_J(\mathbf{\zeta})^2(\mathbf{\zeta}_J)^p\mathbf b_J}{\Gamma_{[1,k], J}(\mathbf{e}; \mathbf{\zeta})}, \qquad \text{if either $c=0$ or $p+l-1<k-l$};\\
 \Delta_{[1,k]}(\mathbf{e}) \cdot\Big(
\sum\limits_{\substack{J \in \binom{[\lfloor \frac N2 \rfloor]}{j}}}  \frac{\Delta_J(\mathbf{\zeta})^2(\mathbf{\zeta}_J)^p\mathbf b_J}{\Gamma_{[1,k], J}(\mathbf{e}; \mathbf{\zeta})} 
+c
\sum\limits_{\substack{J \in \binom{[\lfloor \frac N2 \rfloor]}{j-1}}}  \frac{\Delta_J(\mathbf{\zeta})^2(\mathbf{\zeta}_J)^p\mathbf b_J}{\Gamma_{[1,k], J}(\mathbf{e}; \mathbf{\zeta})} \Big),  \text{if $c\neq 0$ and $p+l-1=k-l$}.
\end{array}
\right.
\end{align}
Then, provided $\D_k^{(l,p)}\neq0$ for $1\leq k\leq N$,   
there exists a unique solution to the inverse problem specified in \autoref{def:ISP}: 
\begin{subequations}
\begin{align}
&&g_{k'}&=\frac{\D_k^{(\frac{k-1}{2},1)}\D_{k-1}^{(\frac{k-1}{2},1)}}
{\mathbf{e}_{[1,k]}\D_k^{(\frac{k+1}{2},0)}\D_{k-1}^{(\frac{k-1}{2},0)}},   &\text{ if k  is odd},  \label{eq:detinversegodd}\\
&&g_{k'}&= \frac{\D_k^{(\frac{k}{2},1)}\D_{k-1}^{(\frac k2 -1,1)}}
{\mathbf{e}_{[1,k]}\D_k^{(\frac k2,0)}\D_{k-1}^{(\frac{k}{2},0)}}, &\text{ if k is even}. \label{eq:detinversegeven}
\end{align}
\end{subequations}
Likewise, 
\begin{subequations}
\begin{align}
h_{k'}&=\frac
{\mathbf{e}_{[1,k-1]}\D_k^{(\frac{k+1}{2},0)}\D_{k-1}^{(\frac{k-1}{2},0)}}{\D_k^{(\frac{k-1}{2},1)}\D_{k-1}^{(\frac{k-1}{2},1)}},   &\text{ if k is odd},  \label{eq:detinversehodd}\\
h_{k'}&= \frac
{\mathbf{e}_{[1,k-1]}\D_k^{(\frac k2,0)}\D_{k-1}^{(\frac{k}{2},0)}}{\D_k^{(\frac{k}{2},1)}\D_{k-1}^{(\frac k2 -1,1)}}, &\text{ if k is even}. \label{eq:detinverseheven}
\end{align}
\end{subequations}
\end{theorem}
\subsection{Second inverse problem: $(g_j, h_j) \Rightarrow X_j$}\label{sec:sIP}
Finally, the relations  (see equation \eqref{eq:ghX})
$$h_j=\abs{m_j}e^{X_j},\quad g_j=\abs{m_j}e^{-X_j}, \quad X_j=x_j+i\omega_j$$ 
imply 
$$X_j=\ln \frac{h_j}{\abs{m_j}}=\ln \frac{\abs{m_j}}{g_j}.$$
Hence we arrive at the inverse formulae 
relating the spectral data and the positions and the momenta of the peakons.  

\begin{theorem} \label{thm:inversex}
Let $W$, given by \autoref{def:ISP}, be the Weyl function for the 
boundary value problem \autoref{eq:dstring} with the  associated spectral data $\{\zeta_j,b_j, c\}$.  
Then the positions $x_j$ and the phases $\omega_j$ (of peakons) in the discrete measure 
$m=2\sum_{j=1}^N \abs{m_j}e^{i\omega_j} \delta_{x_j} $ can be expressed in 
terms of the spectral data as: 

\begin{subequations}
\begin{align}
&&x_{k'}&=\ln \frac
{\mathbf{e}_{[1,k-1]}|\D_k^{(\frac{k+1}{2},0)}||\D_{k-1}^{(\frac{k-1}{2},0)}|}{|m_{k'}||\D_k^{(\frac{k-1}{2},1)}||\D_{k-1}^{(\frac{k-1}{2},1)}|},   &\text{ if k  is odd}, \label{eq:detinversexodd}\\
&&x_{k'}&= \ln \frac
{\mathbf{e}_{[1,k-1]}|\D_k^{(\frac k2,0)}||\D_{k-1}^{(\frac{k}{2},0)}|}{|m_{k'}||\D_k^{(\frac{k}{2},1)}||\D_{k-1}^{(\frac k2 -1,1)}|}, &\text{ if k is even}, \label{eq:detinversexeven}\\
&&e^{i\omega_{k'}}&=\frac
{\D_k^{(\frac{k+1}{2},0)}\D_{k-1}^{(\frac{k-1}{2},0)}}{\D_k^{(\frac{k-1}{2},1)}\D_{k-1}^{(\frac{k-1}{2},1)}}\frac{|\D_k^{(\frac{k-1}{2},1)}||\D_{k-1}^{(\frac{k-1}{2},1)}|}{|\D_k^{(\frac{k+1}{2},0)}||\D_{k-1}^{(\frac{k-1}{2},0)}|},  &\text{ if k  is odd}, \label{eq:detinversemodd}\\
&&e^{i\omega_{k'}}&=\frac
{\D_k^{(\frac k2,0)}\D_{k-1}^{(\frac{k}{2},0)}}{\D_k^{(\frac{k}{2},1)}\D_{k-1}^{(\frac k2 -1,1)}}\frac{|\D_k^{(\frac{k}{2},1)}||\D_{k-1}^{(\frac k2 -1,1)}|}{|\D_k^{(\frac k2,0)}||\D_{k-1}^{(\frac{k}{2},0)}|}|m_{k'}|,  &\text{ if k  is even}, \label{eq:detinversemeven}
\end{align}
\end{subequations}
with $\D_k^{(l,p)}$ defined in \eqref{eq:D}, $k'=N-k+1, \, 1\leq k\leq N$ 
and the convention that $\D_0^{l,p}=1$.  
\end{theorem}

\begin{example}[1-peakon solution] \label{ex:1peakon}
$$X_1=\ln \frac{h_1}{\abs{m_1}}, $$
where
  \[
  h_1=c.
  \]
\end{example}
This case does not require the inverse spectral machinery.

\begin{example}[2-peakon solution]\label{ex:2peakon}
$$X_j=\ln \frac{h_j}{\abs{m_j}},\qquad j=1,2,$$
where
  \begin{align}\label{eq:inverse2peakons}
     h_1=\frac{b_1}{\zeta_1(1+\zeta_1|m_2|^2)},\ \ \
     h_2=\frac{b_1|m_2|^2}{1+\zeta_1|m_2|^2}, \ \ \ b_1(t)=b_1(0)e^{\frac{2t}  {\zeta_1}}.  
       \end{align}
\end{example}
Observe that 
$X_2-X_1=\ln \abs{m_1} \abs{m_2}\zeta_1$ so 
both the distance between the peakons and their relative phases are 
constant in time (see \autoref{fig:hp2peakon1}).  
Note, however, that to respect the ordering 
$x_1<x_2$, 
$$\frac{1}{\abs{\zeta_1}}< \abs{m_1}\abs{m_2}$$ 
must hold.  
\begin{figure}[ht!]
  \centering
  \resizebox{1\textwidth}{!}{
 \includegraphics{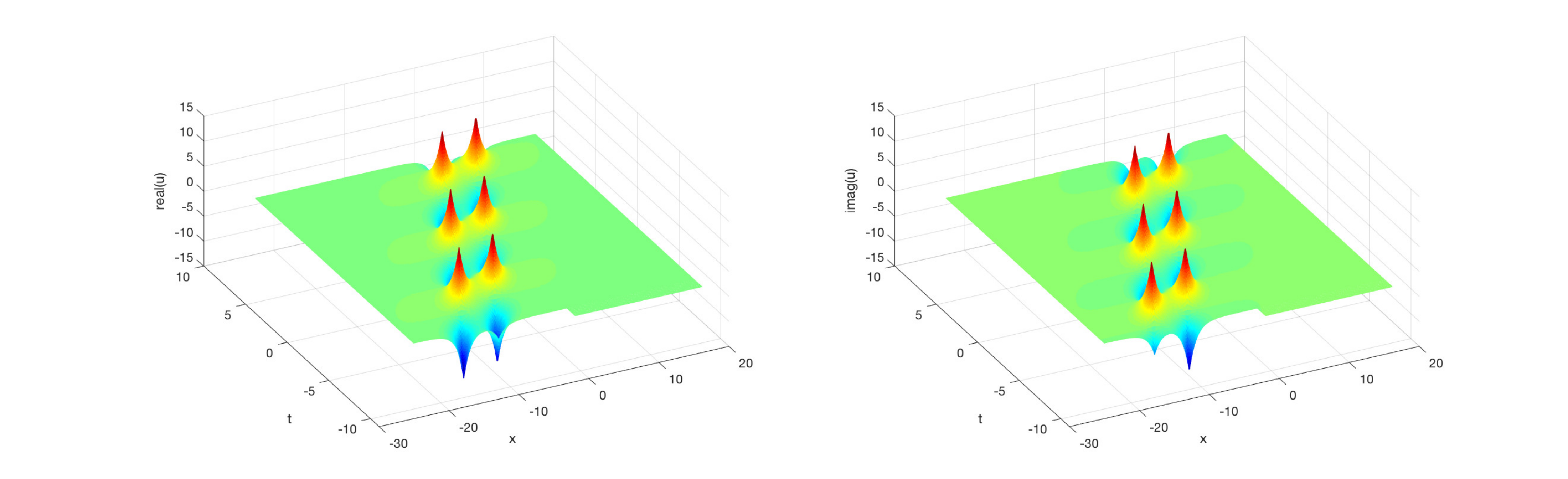}
  }
   \caption{2-peakon solution; $b_1(0)=2+i,\ \zeta_1=1+i,\ |m_1|=10,\ |m_2|=10$. Peakons form a bound state. }
   \label{fig:hp2peakon1}
\end{figure}

This inequality can also be arranged to hold in particular if $\zeta_1$ is 
purely imaginary.  Thus there exist two-peakon 
breather solutions (see \autoref{fig:hp2peakon1_periodic}). 
 
\begin{figure}[ht!]
  \centering
  \resizebox{1\textwidth}{!}{
  \includegraphics{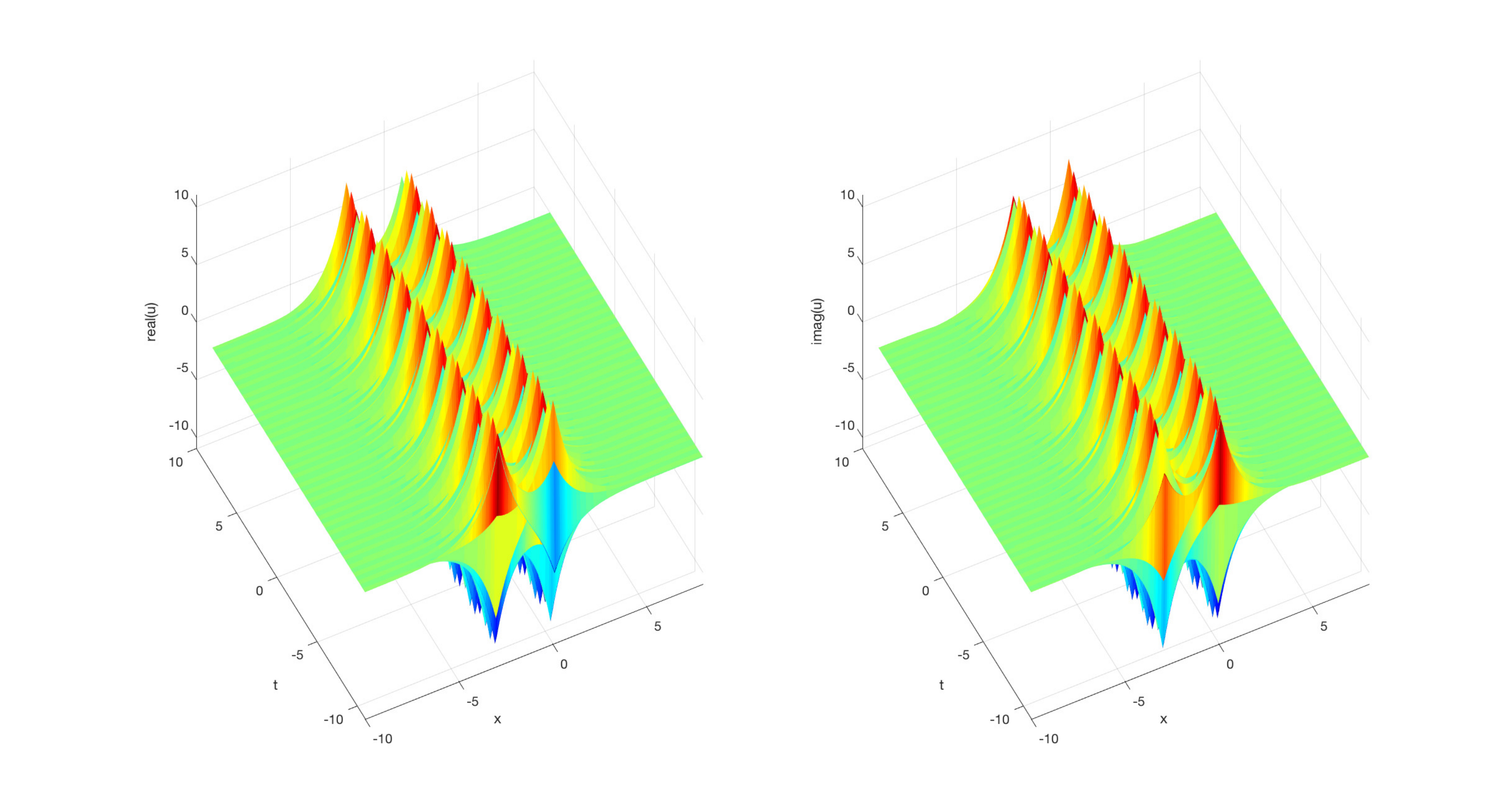}
  }
\caption{Periodic 2-peakon breather; $ b_1(0)=2+i,\ \zeta_1=0.2i,\ |m_1|=10,\ |m_2|=10$. }
   \label{fig:hp2peakon1_periodic}
\end{figure}

\begin{example}[3-peakon solution]\label{ex:3peakon}
$$X_j=\ln \frac{h_j}{\abs{m_j}}, \qquad b_1(t)=b_1(0)e^{\frac{2t}{\zeta_1}},   \qquad  j=1,2,3,$$
where
\begin{subequations}
  \begin{align}\label{eq:inverse3peakons}
     &h_1=\frac{b_1c}{\zeta_1\left(b_1\zeta_1|m_2|^2|m_3|^2+c(1+\zeta_1|m_2|^2)(1+\zeta_1|m_3|^2)\right)},&\\
     &h_2=\frac{b_1|m_2|^2}{b_1\zeta_1|m_2|^2|m_3|^2+c(1+\zeta_1|m_2|^2)(1+\zeta_1|m_3|^2)}\left(\frac{b_1|m_3|^2}{1+\zeta_1|m_3|^2}+c\right),&\\
     &h_3=\frac{b_1|m_3|^2}{1+\zeta_1|m_3|^2}+c.
  \end{align}
  \end{subequations} 
The formulas \eqref{eq:inverse3peakons} are local formulas, nevertheless we can prove, under appropriate conditions on the initial data, that there always exists a 
global (i.e. valid for arbitrary time $t\in (t_0, +\infty)$)  $3$-peakon solution, originating at some initial time $t_0$. 
\begin{theorem} \label{thm:3peakoneers} 
Suppose that  $$\re(\zeta_1)>0, \qquad \frac{1}{\abs{\zeta_1}}<\abs{m_2}\abs{m_3}. $$ 
Then there always exists a choice of $b_1(0)$ for which 
the solutions $X_j, j=1, 2,3$ are global.  
\end{theorem} 
\begin{proof} 
By construction if all $h_j, j=1,2,3$ are well defined 
and the ordering condition $x_1<x_2<x_3$ holds then 
$X_1, X_2,X_3$ satisfy \eqref{eq:pHPcplx}.  Let us 
consider $b_1(t)=b_1(0) e^{\frac{2t}{\zeta_1}}$ with fixed, but 
otherwise arbitrary, $b_1(0)$.  
Since $\re(\zeta_1)>0$ the modulus of $b_1(t)$ can be made 
arbitrary large by choosing $t$ large enough.  Thus for 
$t$ large enough none of the denominators 
in $h_1, h_2$ can become $0$ while the denominator of $h_3$ 
is never $0$ in view of the assumption on $\zeta_1$.  
We observe that the ordering condition $x_1< x_2< x_3$ can be 
stated 
\begin{equation}\label{eq:hordering}
\abs{\frac{h_1}{m_1}}<\abs{\frac{h_2}{m_2}}<\abs{\frac{h_3}{m_3}}. 
\end{equation}
Let us now analyze the ordering condition in the region 
$t\rightarrow \infty$.  We have 
\begin{align} \label{eq:asympt hs}
&\abs{\frac{h_1}{m_1}}\approx \frac{\abs{c}}{\abs{\zeta_1}^2\abs{m_1}\abs{m_2}^2\abs{m_3}^2}, \\
&\abs{\frac{h_2}{m_2}}\approx \frac{\abs{b_1}}{\abs{\zeta_1}\abs{m_2}\abs{(1+\zeta_1\abs{m_3}^2)}}, \\
&\abs{\frac{h_3}{m_3}}\approx\frac{\abs{b_1}\abs{m_3}}{\abs{(1+\zeta_1\abs{m_3}^2)}}, 
\end{align} 
and in that region, enforcing \eqref{eq:asympt hs}, one is led to 
\begin{align} 
\frac{\abs{c(1+\zeta_1 \abs{m_2}^2)}}{\abs{\zeta_1}\abs{m_1}\abs{m_2} \abs{m_3}^2}< \abs{b_1}, \label{eq:bconstraint}\\
\frac{1}{\abs{\zeta_1}}<\abs{m_2}\abs{m_3} \label{eq:zetaconstraint}, 
\end{align}
the first holding in the asymptotic region without any further assumptions, the second holding 
in view of the assumption on $\zeta_1$.  However, 
in view of continuity in $t$ we can extend the asymptotic inequalities to 
a region $[t_0,\infty )$, for some $t_0\geq 0$, without violating 
the inequalities \eqref{eq:hordering}.  Now it suffices to choose 
a new $\tilde b_1(0)=b_1(0) e^{\frac{2t_0}{\zeta_1}}$ and construct 
$h_1, h_2, h_3$ using the inverse formulas \eqref{eq:inverse3peakons}. 
The resulting $x_1,x_2, x_3$ will by construction satisfy the ordering condition for 
arbitrary $t>0$.  
\end{proof} 
\begin{figure}[ht!]
  \centering
  \resizebox{1.1\textwidth}{!}{
  \includegraphics{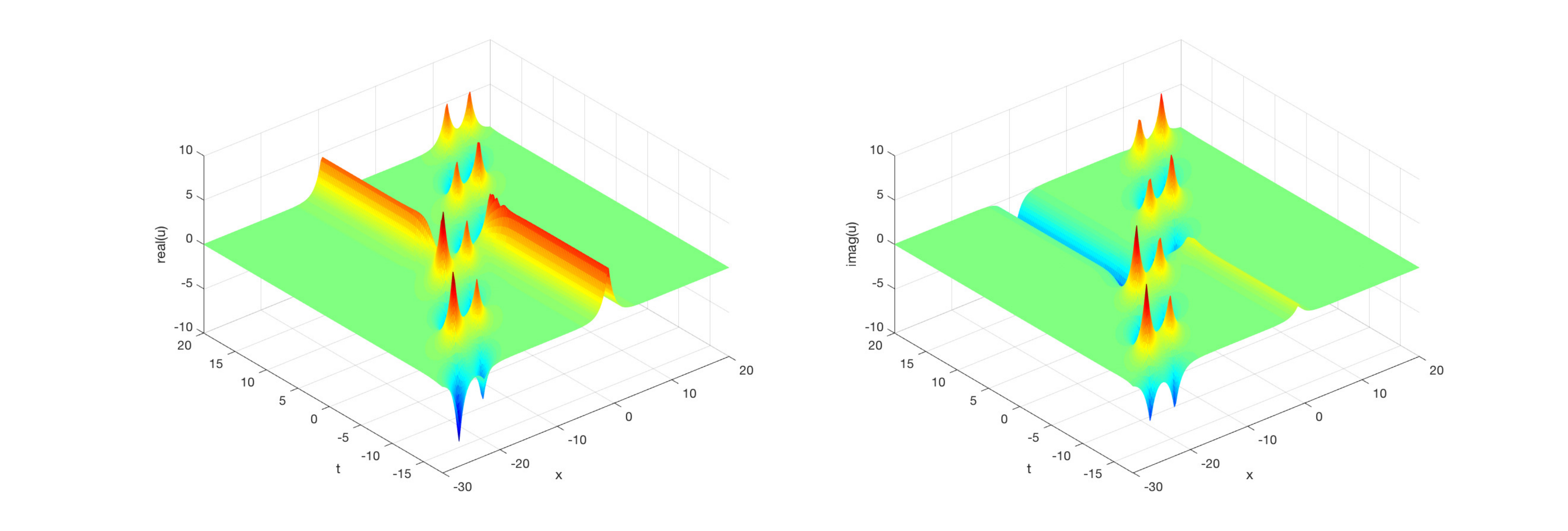}
  }
   \caption{3-peakon solution; $ b_1(0)=1+i,\ c=2+0.5i,\ \zeta_1=1.5+i,\ |m_1|=8,\ |m_2|=5,\ |m_3|=6$}
   \label{fig_hp3peakon1}
\end{figure}

Once the existence of global solutions is established the asymptotic 
behaviour follows from explicit formulas \eqref{eq:inverse3peakons}.  
In particular we see that, asymptotically, peakons pair up as illustrated 
by \autoref{fig_hp3peakon1} which shows the 3 dimensional 
evolution of profiles $\re (u), \im (u)$ while \autoref{fig:HP3overall} 
illustrates the evolution of $\re (u)$ and $\im (u)$ relative to the graph of 
$\abs{u}$.  The interaction between peakons is 
best captured through the graphs of their trajectories (see \autoref{fig:HP3positions}).   
\begin{corollary} 
Suppose that $b_1(0)$ and $\zeta_1$ satisfy conditions 
of \autoref{thm:3peakoneers}.  
Then the first peakon stops in the asymptotic region, while the second and third 
form a bound pair moving with speed $\frac{2\re \zeta_1}{\abs{\zeta_1}^2}$. 
\end{corollary} 
\begin{figure}[ht!]
  \centering
  \resizebox{1.1\textwidth}{!}{
  \includegraphics{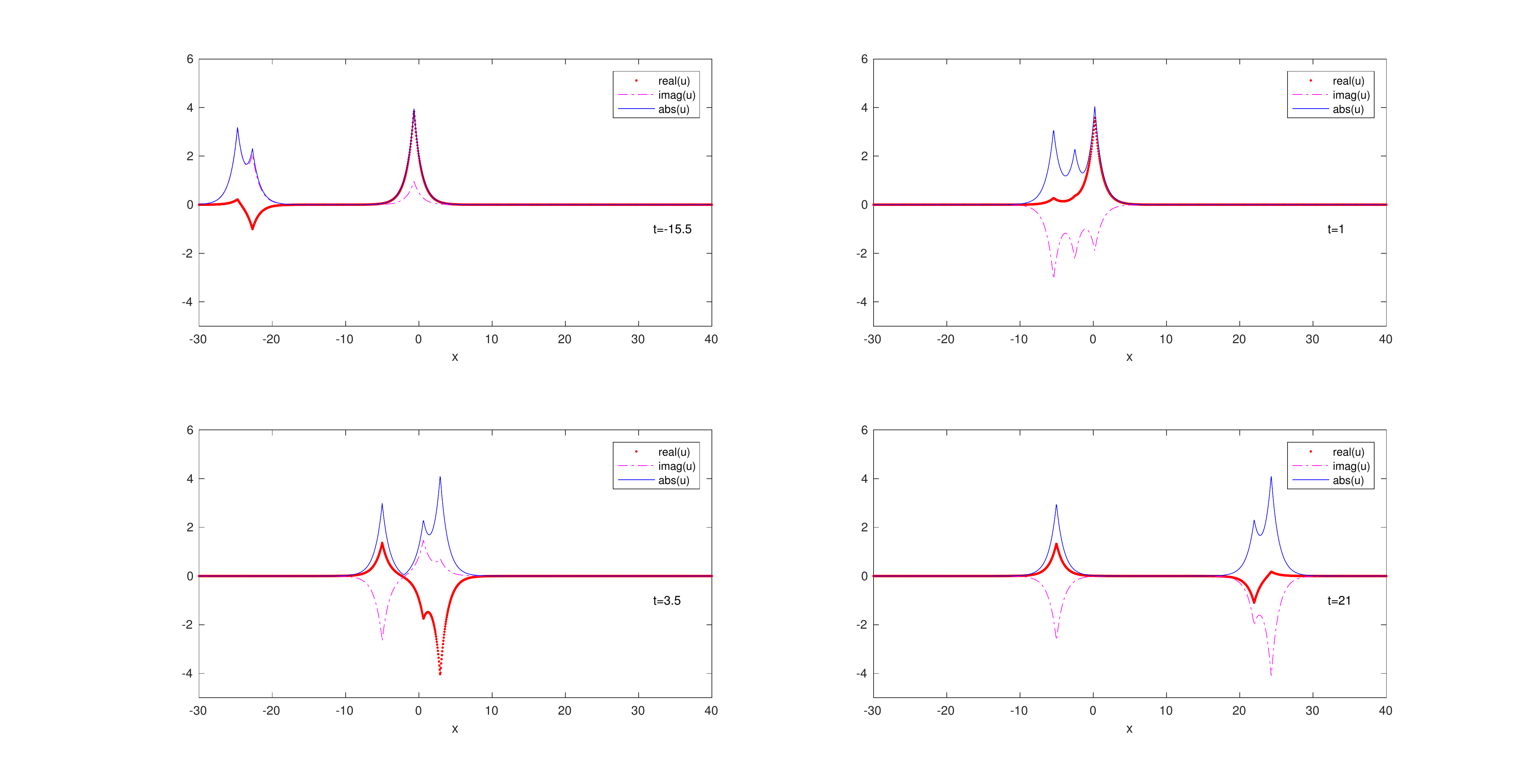}
  }
   \caption{A superimposed view of the amplitude $\abs{u}$  as 
   well as $\re u$ and $\im u$.  Asymptotic 
   pairing is visible in graphs of all these three quantities.}
   \label{fig:HP3overall}
\end{figure}
\begin{figure}[ht!]
  \centering
  \resizebox{0.7\textwidth}{!}{
  \includegraphics{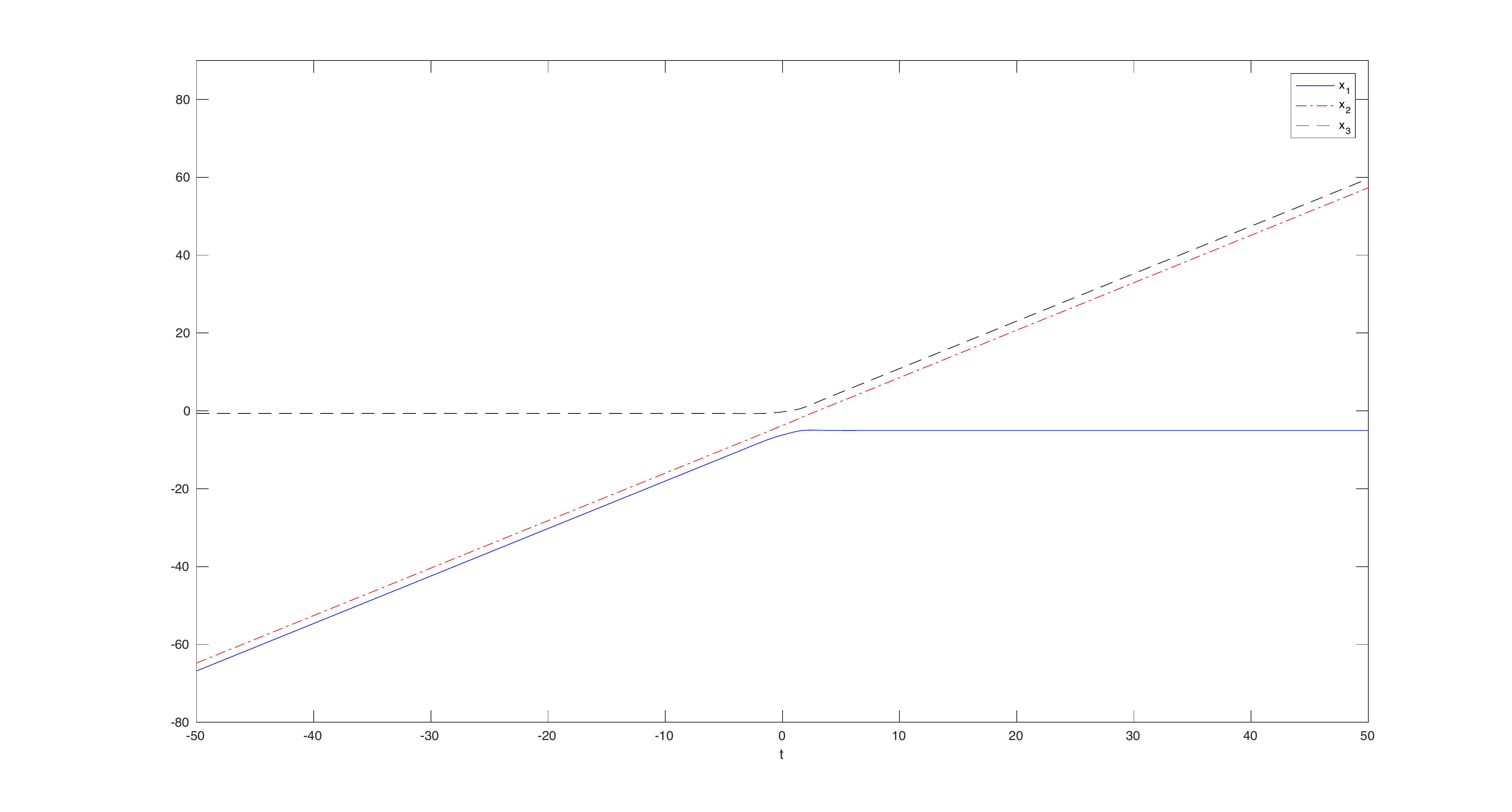}
  }
   \caption{Asymptotic 
   pairing of positions $x_2,x_3$ for large positive times. 
    Observe that initially the first and the second peakons form a bound state; the interaction at $t=0$ with the 
   third peakon breaks the bond and a new bond emerges, 
   while the first peakon slows to a halt.}
   \label{fig:HP3positions}
\end{figure}
We also point out that this analysis can be carried out for 
purely imaginary $\zeta_1$ by forcing $b_1(0)$ to be large 
enough to satisfy \eqref{eq:bconstraint} while at the same time imposing 
\eqref{eq:zetaconstraint}; the resulting $3$-peakon breather is 
graphed in \autoref{fig_hp3peakon_peri}.  

\begin{figure}[ht!]
  \centering
  \resizebox{1.1\textwidth}{!}{
  \includegraphics{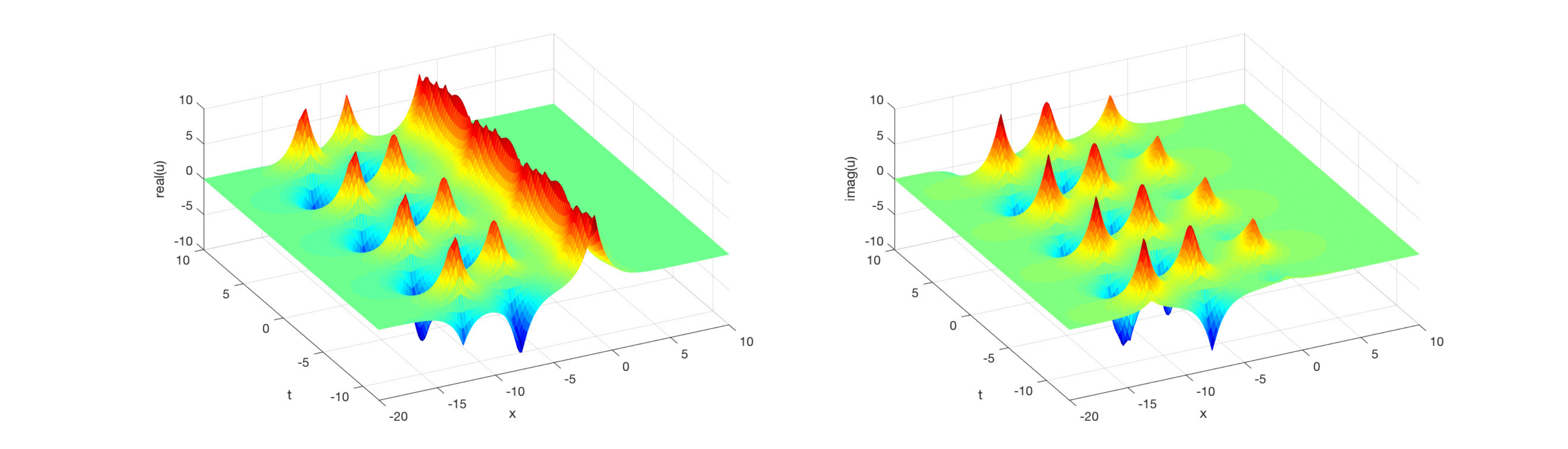}
  }
   \caption{A 3-peakon breather; Re(u) and Im(u) graphed for 
  $ b_1(0)=1+i,\ ~c=~2+0.5i,\ \zeta_1=2i,\ |m_1|=9,\ |m_2|=8,\ |m_3|=10$ }
   \label{fig_hp3peakon_peri}
\end{figure}
\end{example}
The generalization of the analysis of the 3-peakon solutions done above 
to multipeakons will be carried out elsewhere.  For now, we 
confine ourselves to stating the explicit formula for 4-peakons.  
\begin{example}[4-peakon solution]\label{ex:4peakon}
$$X_j=\ln \frac{h_j}{\abs{m_j}}\,\qquad b_1(t)=b_1(0)e^{\frac{2t}{\zeta_1}}, 
\qquad b_2(t)=b_2(0)e^{\frac{2t}{\zeta_2}}, \qquad j=1,2,3,4,$$
where
  \begin{align*}
     &h_1=\frac{b_1b_2(\zeta_2-\zeta_1)^2}{\zeta_1\zeta_2\left(b_1\zeta_1(1+\zeta_2|m_2|^2)(1+\zeta_2|m_3|^2)(1+\zeta_2|m_4|^2)+b_2\zeta_2(1+\zeta_1|m_2|^2)(1+\zeta_1|m_3|^2)(1+\zeta_1|m_4|^2)\right)},&\\
     &h_2=|m_2|^2\cdot\frac{b_1b_2(\zeta_2-\zeta_1)^2\left(b_1(1+\zeta_2|m_3|^2)(1+\zeta_2|m_4|^2)+b_2(1+\zeta_1|m_3|^2)(1+\zeta_1|m_4|^2)\right)}{\left(b_1\zeta_1(1+\zeta_2|m_3|^2)(1+\zeta_2|m_4|^2)+b_2\zeta_2(1+\zeta_1|m_3|^2)(1+\zeta_1|m_4|^2)\right)}&\\
     &\qquad\qquad \cdot\left.\frac{1}{\left(b_1\zeta_1(1+\zeta_2|m_2|^2)(1+\zeta_2|m_3|^2)(1+\zeta_2|m_4|^2)+b_2\zeta_2(1+\zeta_1|m_2|^2)(1+\zeta_1|m_3|^2)(1+\zeta_1|m_4|^2)\right)}\right),&\\
     &h_3=\frac{\left(b_1(1+\zeta_2|m_4|^2)+b_2(1+\zeta_1|m_4|^2)\right)\left(b_1(1+\zeta_2|m_3|^2)(1+\zeta_2|m_4|^2)+b_2(1+\zeta_1|m_3|^2)(1+\zeta_1|m_4|^2)\right)}{(1+\zeta_1|m_4|^2)(1+\zeta_2|m_4|^2)\left(b_1\zeta_1(1+\zeta_2|m_3|^2)(1+\zeta_2|m_4|^2)+b_2\zeta_2(1+\zeta_1|m_3|^2)(1+\zeta_1|m_4|^2)\right)},&\\
     &h_4=|m_4|^2\cdot\frac{b_1(1+\zeta_2|m_4|^2)+b_2(1+\zeta_1|m_4|^2)}{(1+\zeta_1|m_4|^2)(1+\zeta_2|m_4|^2)}.&
       \end{align*}
       \end{example}
 
We finish this section by briefly 
commenting about the NLSP equation in the conservative 
peakon sector.  As we pointed out in \autoref{lem:HP-NLSP duality}
the peakon flows for HP and NLSP form an orthogonal family 
and our analysis using the inverse spectral problem carries over to 
the NLSP case.  

We recall that \autoref{eq:NLSP} 
is the compatibility condition of 
\begin{equation} \label{eq:NLSPLax}
\Psi_x=U \Psi, \quad  \Psi _t =V \Psi, \quad  \Psi=\begin{bmatrix} \Psi_1\\\Psi_2 \end{bmatrix} ,
\end{equation} 
where
\begin{align*} 
&U=\frac{1}{2}\begin{bmatrix} -1 &\lambda m\\ -\lambda \bar{m}& 1 \end{bmatrix},\\
\\ 
&V=\frac{i}{2}\begin{bmatrix} 4\lambda^{-2} + Q & -2\lambda^{-1} (u-u_x)-\lambda m i \im(Q)\\
2\lambda^{-1}(\bar u+\bar u_x)+\lambda \bar m i \im(Q)& -Q \end{bmatrix},
\end{align*} 
with $Q$ given by \eqref{eq:Q}. 

The spectral problem is the same, namely given by 
\autoref{eq:xLaxBVP}.  
 The only point of departure from the HP case 
is the time evolution of the spectral date which 
can be obtained by using $V$ from \autoref{eq:NLSPLax} 
in the asymptotic region, resulting in a modest variation of 
\autoref{lem:t-evolution of qp}.  In summary,  the NLSP time evolution of $W$ is: 
\begin{equation}\label{eq:dotWNLSP}
\dot W=i\big(\frac2z\, W-\frac{2L}{z}\big), 
\end{equation}
which in turn leads to the NLSP time evolution of 
the spectral data given by the following theorem (see 
\autoref{eq:tflowSD} for comparison).  
\begin{lemma}
Suppose $p_N(z)$ has simple roots.  Then in the notation of 
\autoref{eq:simpleW} 
 the spectral data $\{\zeta_j,b_j, c\}$ evolve according to
\begin{align}\label{eq:tflowSDNLSP}
\dot\zeta_j=0,\qquad \dot b_j=\frac{2i}{\zeta_j}b_j,\qquad \dot c=0.
\end{align}
\end{lemma}

Finally, the generalization to the $\theta$ family given by \eqref{eq:pmfamily} 
is straightforward and we only mention that 
the Weyl function evolves according to: 
\begin{equation}\label{eq:dotWtheta}
\dot W=e^{\inum \theta} \big(\frac2z\, W-\frac{2L}{z}\big), 
\end{equation}
implying that the spectral data $\{\zeta_j,b_j, c\}$ evolves 
\begin{align}\label{eq:tflowmfamily}
\dot\zeta_j=0,\qquad \dot b_j=\frac{2e^{\inum \theta} }{\zeta_j}b_j,\qquad \dot c=0.
\end{align}
\section{Conclusions} \label{sec: Conclusions}
We showed that the NLS-type peakon equations introduced 
in \cite{anco} can be generalized to a family of peakon equations 
parametrized by the real projective line.  We studied the sector of 
conservative peakon solutions for which we formulated and solved 
the inverse problem resulting in explicit peakon solutions of various types 
such as peakons and periodic peakon breathers.  This work opens multiple 
paths to further studies.  Some of the outstanding issues are: 
\begin{enumerate} 
\item the Poisson structures for the conservative peakon sector 
were obtained by analyzing the (singular) limit of the Hamiltonian structures valid in the smooth case, followed by, what amounts to, guessing the right structure and it is of general interest to 
understand more deeply which Hamiltonian structures survive that singular limit; 

\item the inverse problem was solved under the assumption of 
simple eigenvalues and it is not known to us what 
dynamical consequences will result from lifting of that assumption; 

\item the analysis of the global (in $t$) existence of solutions was only done for a small number of peakons and a generalization is called for; 

\item we have done no stability analysis of conservative peakons and this 
remains an important open problem.

\end{enumerate} 

\section{Acknowledgements} 
X. Chang was supported in part by the Natural  Sciences and Engineering Research Council of Canada (NSERC), the Department of Mathematics and Statistics of the University of Saskatchewan, PIMS postdoctoral fellowship and the National Natural Science Foundation of China (No. 11701550, No. 11731014). 
S. Anco and J. Szmigielski are supported in part by NSERC.

\begin{appendices}
\section{Some results for the conservative peakon sector}\label{sec:AppPB}

In this Appendix we collect some formulas used in \autoref{sec:peakonPB} 
for the discussion of Poisson structures and for the Hamiltonian form of 
the peakon flows \autoref{eq:pHPcplx} and \autoref{eq:pNLSPcplx}.  

Our main goal is to express $H^{(0)}$ and $E^{(0)}$ 
in terms of the complex variable $X_j$ introduced in \eqref{eq:Xj}.  
Throughout we work with the ordering $x_1<x_2<\cdots<x_N$.  
\begin{lemma}\label{lem:H0}
\begin{equation}
H^{(0)}=||u||_{H^1}^2=4 \re\big(\sum_{k<l} \abs{m_k}\abs{m_l}
e^{X_k-X_l}\big)+2\sum_l \abs{m_l}^2.  
\end{equation} 
\end{lemma} 
\begin{proof} 
We recall (see \autoref{eq:H0}) 
\begin{equation*} 
\begin{split} 
H^{(0)}=&\re \int \bar u m dx=2 \re \big(
\sum_{k,l} m_k \bar m_l e^{-\abs{x_k-x_l}}\big)=
2\re \big( \sum_{k<l} m_k \bar m_l e^{x_k-x_l}+\sum_{k} \abs{m_k}^2+ 
\\ &\sum_{l<k} m_k \bar m_l e^{x_l-x_k}\big)=2\re \big( 
\sum_{k<l} \abs{m_k} \abs{m_l} e^{X_k-X_l}+\sum_{k}\abs{m_k}^2 + 
\sum_{l<k} \abs{m_k} \abs{m_l} e^{\bar X_l-\bar X_k}\big)=\\
&4 \re \big(\sum_{k<l} \abs{m_k} \abs{m_l} e^{X_k-X_l}\big)+ 2\sum_k
\abs{m_k}^2 . 
\end{split} 
\end{equation*} 
\end{proof} 
\begin{lemma}\label{lem:E0}
\begin{equation}
E^{(0)}=\im \int{u_x \bar m} dx=-4\im(\sum_{k<l} \abs{m_k}\abs{m_l}
e^{X_k-X_l})
\end{equation}
\end{lemma} 
\begin{proof} 
By definition 
\begin{align*} 
&\im \int{u_x \bar m} dx=2\im \sum_l\avg{u_x}(x_l) \bar m_l=
2\im \sum_{k,l} m_k\bar m_l \sgn(x_k-x_l) e^{-\abs{x_l-x_k}}=\\
&2\im \big(-\sum_{k<l}\abs{m_k}\abs{m_l} e^{X_k-X_l}+\im \sum_{l<k}\abs{m_k}\abs{m_l} e^{\bar X_l-\bar X_k}\big)=\\
&2\im \big(-\sum_{k<l}\abs{m_k}\abs{m_l} e^{X_k-X_l}+ \sum_{k<l}\abs{m_k}\abs{m_l} e^{\bar X_k-\bar X_l}\big)=-4 \im \big(\sum_{k<l}\abs{m_k}\abs{m_l} e^{X_k-X_l}\big). 
\end{align*} 

\end{proof} 
\begin{lemma} \label{lem:Q} 
Suppose $x\notin supp(m)$ 
then 
\begin{equation} 
Q(x)=4 \sum_{x_k<x<x_l} 
\abs{m_k} \abs{m_l} e^{X_k-X_l}, 
\end{equation} 
\end{lemma} 
\begin{proof} 
We recall that by \autoref{eq:Q} 
\begin{equation*} 
Q(x)=(u-u_x)(\bar u+\bar u_x). 
\end{equation*} 
For the peakon Ansatz \eqref{eq:peakonansatz}
\begin{equation*} 
\begin{split}
u-u_x=&\sum_k m_k(1-\sgn(x_k-x))e^{-\abs{x-x_k}}=2\sum_{x_k<x} 
m_k e^{-\abs{x-x_k}}=\\
&\big(2\sum_{x_k<x} \abs{m_k} e^{X_k}\big) e^{-x}, 
\end{split} 
\end{equation*}
and 
\begin{equation*}
\begin{split} 
\bar u+\bar u_x=&\sum_l \bar m_l(1+\sgn(x_l-x))e^{-\abs{x-x_l}}=2\sum_{x<x_l} 
\bar m_l e^{-\abs{x-x_l}}=\\
&\big(2\sum_{x<x_l} \abs{m_l} e^{-X_l}\big) e^{x},
\end{split} 
\end{equation*}
thus proving the claim.  
\end{proof}
\begin{lemma} \label{lem:Qavg} 
Let $x_j \in supp(m)$.  Then 
\begin{equation*}
\avg{Q}(x_j)=4 \sum_{k<j<l} \abs{m_k} \abs{m_l} e^{X_k-X_l}+
2\abs{m_j} \big(\sum_{k<j} \abs{m_k}e^{X_k-X_j}+ \sum_{j<k} 
\abs{m_k} e^{X_j-X_k}\big).  
\end{equation*} 
\end{lemma} 
\begin{proof} A straightforward computation yields
\begin{equation*}
\begin{split} 
&\avg{Q}(x_j)=\frac12\big( Q(x_j+)+Q(x_j-)\big)\stackrel{\autoref{lem:Q}}{=}
2\big(\sum_{x_k<x_j+<x_l }\abs{m_k}\abs{m_l} e^{X_k-X_l}+
\sum_{x_k<x_j-<x_l} \abs{m_k}\abs{m_l} e^{X_k-X_l}\big)=\\
&2\big(\sum_{k\leq j<l} \abs{m_k}\abs{m_l} e^{X_k-X_l}+ \sum_{k<j\leq l} \abs{m_k}\abs{m_l} e^{X_k-X_l}\big)=\\&4\sum_{k<j<l} \abs{m_k} \abs{m_l} 
e^{X_k-X_l}+ 2\abs{m_j}\big(\sum_{k<j} \abs{m_k} e^{X_k-X_j}+\sum_{j<k} 
\abs{m_k} e^{X_j-X_k}\big).  
\end{split}
\end{equation*}
\end{proof}
\end{appendices}

\bibliographystyle{abbrv}
\bibliography{NLSH}
\end{document}